\documentclass[hidelinks,11pt]{article} 
\usepackage{fullpage}
\usepackage{microtype}
\usepackage{graphicx}
\usepackage{subfig}
\usepackage{booktabs} 
\usepackage{hyperref}
\usepackage{xcolor}
\hypersetup{
    colorlinks=true,
    linkcolor=red,
    citecolor=blue,
    filecolor=magenta,      
    urlcolor=cyan,
   }     

\usepackage{algorithm}
\usepackage{algpseudocode}

\usepackage[utf8]{inputenc} 
\usepackage[T1]{fontenc}    
\usepackage{url}           
\usepackage{amsfonts}       
\usepackage{nicefrac}       
\usepackage{amsmath,amssymb,amsthm}
\usepackage{color}
\usepackage{mathtools}
\usepackage{enumitem}
\usepackage{epstopdf}
\usepackage{tabulary}
\usepackage{cite}
\usepackage{chngcntr}
\usepackage{bm}
\usepackage[algo2e,linesnumbered,ruled,vlined,resetcount]{algorithm2e}
\RestyleAlgo{ruled}

\SetCommentSty{mycommfont}

\newtheorem{theorem}{Theorem}
\newtheorem{lemma}{Lemma}
\newtheorem{proposition}{Proposition}

\newtheorem{assumption}{Assumption}

\newtheorem{definition}{Definition}
\newtheorem{remark}{Remark}

\newcommand{\blue}[1]{{\color{blue}{#1}}}

\DeclareMathOperator{\Tr}{\mathrm{tr}}

\title{Model-Free Output Feedback Stabilization via  Policy Gradient Methods}
\author{Ankang Zhang, Ming Chi, Xiaoling Wang, and Lintao Ye
\thanks{
A. Zhang, M. Chi, and L. Ye are with the School of Artificial Intelligence and Automation, Huazhong University of Science and Technology, Wuhan 430074, China, email: \{zak, chiming, yelintao93\}@hust.edu.cn. X. Wang is with the College of Automation, Nanjing University of Posts and Telecommunications, Nanjing 210023, China, email: xiaolingwang@njupt.edu.cn. } 
}
\begin{document}
\maketitle

\begin{abstract}
Stabilizing a dynamical system is a fundamental problem that serves as a cornerstone for many complex tasks in the field of control systems. The problem becomes challenging when the system model is unknown. Among the Reinforcement Learning (RL) algorithms that have been successfully applied to solve problems pertaining to unknown linear dynamical systems, the policy gradient (PG) method stands out due to its ease of implementation and can solve the problem in a model-free manner. However, most of the existing works on PG methods for unknown linear dynamical systems assume full-state feedback. In this paper, we take a step towards model-free learning for partially observed linear dynamical systems with output feedback and focus on the fundamental stabilization problem of the system. We propose an algorithmic framework that stretches the boundary of PG methods to the problem without global convergence guarantees. We show that by leveraging zeroth-order PG update based on system trajectories and its convergence to stationary points, the proposed algorithms return a stabilizing output feedback policy for discrete-time linear dynamical systems. We also explicitly characterize the sample complexity of our algorithm and verify the effectiveness of the algorithm using numerical examples.
\end{abstract}

\section{Introduction}
Stabilizing dynamical systems is a foundational and enduring topic in control theory. In both theoretical and practical aspects, identifying an initial stabilizing controller is a critical prerequisite for solving more complex optimization tasks, such as the linear quadratic regulator (LQR) problem. While such a task is often straightforward in model-based settings, where system model is learned from data or given prior~\cite{dean2020sample}, it becomes non-trivial in model-free settings, where system models are absent. However, model-free approaches are particularly attractive as they bypass the system identification phase, thereby avoiding potential modeling errors~\cite{hou2013model,recht2019tour}. As a staple of reinforcement learning, policy gradient (PG) methods have achieved significant milestones for model-free control, particularly for solving LQR problems in both continuous-time and discrete-time linear systems~\cite{fazel2018global,mohammadi2019global,mohammadi2020linear,fatkhullin2021optimizing,mohammadi2021convergence,malik2020derivative,cassel2021online}. 
Although both the LQR cost and the set of stabilizing policies are non-convex~\cite{fazel2018global,fatkhullin2021optimizing}, the seminal work in~\cite{fazel2018global,mohammadi2019global} demonstrated that the LQR cost exhibits gradient dominance (i.e., Polyak-Łojasiewicz condition~\cite{karimi2016linear}), thereby guaranteeing the global convergence of PG methods. Building on this theoretical foundation, subsequent studies have extended PG methods to various LQR formulations, such as Markovian jump LQR~\cite{jansch2022policy}, risk-constrained LQR~\cite{zhao2023global}, and LQR over rate-limited channels~\cite{ye2024model}.
However, these results rely on a common assumption that an initial stabilizing controller is available. 

There is a line of research based on the so-called discount methods to stabilize unknown control systems~\cite{perdomo2021stabilizing,ozaslan2022computing,zhao2024convergence}.
The core mechanism hinges on relaxing the stabilization of the original system into a sequence of discounted LQR problems on damped systems. This allows us to start from a trivial initial policy (e.g., a zero policy) and return a stabilizing controller for the original system as the discount factor gradually approaches $1$.
However, the aforementioned works predominantly address the full state feedback (SF) scenario, where the control input is computed using system state.
In practical control applications, acquiring the complete state information of a system can be prohibitively expensive or even impossible~\cite{qian2024robust, Chi10935305,duan2023optimization1}. 
This motivates us to study the problem in this work: \textit{Design a model-free algorithm for learning stabilizing output feedback controllers and analyze its sample complexity}.

While the model-free PG methods have been applied to solve the linear quadratic control problem of partially observed systems (i.e., when only the system output is available), the existing works have mainly focused on the optimization landscape and convergence of PG methods to stationary points~\cite{fatkhullin2021optimizing,duan2022optimization,duan2023optimization1,duan2023optimization2,tang2023analysis}, starting from a known stabilizing policy. These works neither provide explicit sample complexity analysis nor address the problem of learning a stabilizing output feedback controller.
Unlike the SF scenario, the output feedback setting is structurally more challenging as it lacks the gradient dominance property, and the set of stabilizing policies may be disconnected~\cite{fatkhullin2021optimizing,bu2020topological}. Therefore, the global convergence analysis commonly used in the SF setting cannot be directly applied to the output feedback setting.
Specifically, when the system model is known, the analysis in~\cite{fatkhullin2021optimizing,duan2022optimization,duan2023optimization1} provides convergence guarantees for gradient-based optimization algorithms toward stationary points of the cost function.
Another challenge of the output feedback problem is that the system can be stabilized by a dynamic output feedback controller if not by a static output feedback controller~\cite{duan2023optimization2,tang2023analysis}.
Although dynamic or history-dependent output feedback is more expressive, SOF remains attractive due to its memoryless and low-dimensional implementation, while history-based dimensions typically grow with the history length
~\cite{zhao2023globally,sadamoto2026policy}.
Similarly, the aforementioned works rely on the assumption of a known initial stabilizing output feedback controller.
Given the above challenges and the lack of existing sample complexity results for unknown partially observed systems,
motivated by these recent advances~\cite{perdomo2021stabilizing,zhao2024convergence,ozaslan2022computing}, we take the first step to investigate learning stabilizing Static Output Feedback (SOF) policies for discrete-time linear systems with unknown system models. Our main contributions are summarized as follows.

$\bullet$ {\bf Model-Free SOF Synthesis Framework.} We develop a model-free algorithmic framework for learning stabilizing SOF controllers for open-loop unstable discrete-time linear systems from system trajectories. Since the SOF landscape generally lacks gradient dominance available in the SF setting, our framework is built upon zeroth-order PG methods with only local convergence guarantees.

$\bullet$ {\bf Discount Mechanism for Partially Observed System.} We introduce a discount method into the SOF stabilization problem, which transforms the stabilization of the original system into policy learning problems for a sequence of discounted partially observed systems with properly chosen discount factors. 
Such a mechanism is a building block of our algorithmic framework. Our proposed algorithms only require an upper bound on the spectral radius of the open-loop system to start with, and upper bounds on the norms of unknown system matrices.

 $\bullet$ {\bf Convergence Guarantees and Sample Complexity.} We establish the overall convergence of the proposed zeroth-order PG method to a stabilizing SOF controller for the original system,
 despite the inherent lack of gradient dominance in the SOF landscape. This is achieved by rigorously analyzing the error of two-point zeroth-order gradient estimates used in the PG method, and how the estimation error propagates during the course of the algorithm. Furthermore, we provide a comprehensive quantification of the total sample complexity (i.e., the number of system trajectories) required for the entire stabilization procedure.

We now summarize some related work below.

$\bullet$ {\bf Stabilization via PG Methods.} PG methods circumvent the complexities of model identification by searching directly in the controller parameter space. The work of~\cite{lamperski2020computing} proposes a model-free algorithm based on policy iteration to stabilize deterministic LQR systems, while~\cite{perdomo2021stabilizing} develops the discount method framework for stabilization of both linear and nonlinear systems, both without providing a finite-sample complexity analysis. Along this line, the work of~\cite{ozaslan2022computing} addresses the stabilization problem for continuous-time systems. In the discrete-time setting, the authors in~\cite{zhao2024convergence} provided rigorous theoretical guarantees, demonstrating that the sample complexity is of the order $\mathcal{O}(n^2m)$, where $n$ and $m$ denote the state and input dimensions, respectively.
Different from~\cite{zhao2024convergence}, the work of~\cite{toso2025learning} studies the problem of learning to stabilize policies via an unstable
subspace representation, thereby substantially reducing the sample complexity to $\widetilde{\mathcal{O}}(n_u^2m)$ with $n_u$ being the number of unstable modes. 
Recent studies on stabilization problems via policy iteration~\cite{lamperski2020computing} and PG methods~\cite{perdomo2021stabilizing, zhao2024convergence, ozaslan2022computing,toso2025learning} have predominantly focused on the full-state feedback paradigm. In contrast, our work addresses the stabilization problem for unknown partially observed systems, which is more practically relevant yet structurally more challenging due to the lack of gradient dominance.

$\bullet$ {\bf Stabilization via System Identification.} Another indirect route to finding stabilizing controllers is to first perform system identification from data and then utilize the identified model for controller synthesis. For instance, \cite{hu2022sample} studies learning a stabilizing SF controller for an unknown LTI system from a single trajectory with $\widetilde{\mathcal O}(n_u)$ samples. For partially observed systems,~\cite{zhangstabilizing} employs singular value decomposition on the Hankel matrix to estimate unstable modes, achieving a sample complexity of $\mathcal O(n_u^2)$ in the case of a diagonalizable system matrix. 
Additionally, \cite{talebi2021regularizability} introduces the concept of regularizability to facilitate online regulation via a data-guided procedure when only the input matrix is known. While these approaches demonstrate sample efficiency by exploiting system structures, they generally necessitate explicit system identification or sophisticated subspace recovery. In contrast, our approach directly searches the policy space, circumventing the need for any form of model or subspace estimation.

{\bf Notation.}
For a matrix $A$, let $\Tr(A)$, $\rho(A)$, $\underline{\sigma}(A)$, $\|A\|$ and $\|A\|_F$ be its trace, spectral radius, minimal singular value, spectral norm, and Frobenius norm, respectively. 
$\mathbb{R}^{n\times m}$ and $\mathbb{R}^n$ denote the set of $n\times m$ real matrices and the $n$-dimensional real Euclidean space, respectively. $\mathbb{N}_+$ denotes the set of positive integers. $I_n$ represents the $n \times n$ identity matrix.

\section{Problem Formulation and Preliminaries}
\subsection{Problem Formulation}
Consider a discrete-time linear time-invariant (LTI) system characterized by the following state-space realization:
\begin{equation}\label{eq-original system}
  \begin{aligned}
    x_{t+1} &= Ax_{t} + Bu_t, \\
    y_{t} &= Cx_t,
  \end{aligned}
\end{equation}
where $x_t \in \mathbb{R}^{n}$ is the state vector, $u_t \in \mathbb{R}^{m}$ is the control input, and $y_t \in \mathbb{R}^{p}$ is the output. $A\in \mathbb{R}^{n\times n}$, $B\in \mathbb{R}^{n\times m}$ and $C\in \mathbb{R}^{p\times n}$ are the system matrices.

Our work specifically addresses the scenario where the open-loop system is unstable, i.e., the spectral radius $\rho(A) \geq 1$. Throughout the paper, we make the following standard assumptions regarding the initial state distribution (see, e.g.,~\cite{zhao2024convergence,malik2020derivative}) and the system's structural properties (see, e.g.,~\cite{fatkhullin2021optimizing,duan2022optimization}).

\begin{assumption} \label{ass1}
    The initial state $x_0$ is a random vector drawn from a distribution with zero mean and identity covariance. 
    Furthermore, $x_0$ is assumed to be norm-bounded, satisfying $\|x_0\| \leq d$ for some constant $d > 0$.
\end{assumption}

\begin{remark}
Assumption~\ref{ass1} is mainly used to simplify our theoretical analysis. The boundedness assumption can be relaxed to tailed distributions, e.g., Gaussian distributions, such that the norm bound holds with high probability~\cite{malik2020derivative}.
\end{remark}

\begin{assumption} \label{ass:SOF_existence}
   The set of stabilizing static output feedback policies is non-empty. That is, there exists $K \in \mathbb{R}^{m \times p}$ such that the closed-loop system matrix $A - BKC$ is Schur stable (i.e., $\rho(A - BKC) < 1$).
\end{assumption}

{\bf Objective:} Our primary goal is to synthesize a stabilizing SOF controller of the form $u_t = -Ky_t$. We seek to achieve this through PG methods in a model-free manner, where the system matrices $(A, B, C)$ and state dimension $n$ are unknown, and only system rollouts are available.

\subsection{Discount Method for Output Feedback Control}

To address the challenges posed by open-loop instability, we employ a discount method. We begin by recalling the standard discounted LQR problem under full state feedback (see, e.g.,~\cite{bertsekas2012dynamic}):
\begin{equation} \label{lqr-state}
    \begin{aligned}
        &\min_F~J_{\gamma}(F)=\mathbb{E}_{x_0}\sum_{t=0}^{\infty}\gamma^t
        \left(x_t^\top Qx_t+u_t^\top Ru_t\right),\\
        & \text{subject to}~\eqref{eq-original system}~ \text{and}~u_t=-Fx_t, 
    \end{aligned}
\end{equation}
where $\gamma \in (0, 1)$ is the discount factor, $F\in\mathbb{R}^{m\times n}$ is the state feedback gain, and $Q \succ 0, R \succ 0$ are the cost matrices. Motivated by recent advances~\cite{zhao2024convergence,perdomo2021stabilizing}, we extend the discount method to output feedback control. 
Specifically, for a given SOF policy $K \in \mathbb{R}^{m \times p}$, the performance is evaluated by:
\begin{equation} \label{eq-lq-cost}
  \begin{aligned}
    &J_{\gamma}(K)=\mathbb{E}_{x_0}\sum_{t=0}^{\infty}\gamma^t
        \left(x_t^\top Qx_t+u_t^\top Ru_t\right),\\
    &\text{subject to}~\eqref{eq-original system}~ \text{and}~u_t=-Ky_t.
  \end{aligned}
\end{equation}
In this work, the discounted cost \eqref{eq-lq-cost} is employed primarily as a computational tool to learn a stabilizing SOF policy. Rather than pursuing optimality, we leverage the gradient of $J_\gamma(K)$ as a search direction to steer the controller $K$ into the stabilizing region. 

Next, we show that the discounted cost in \eqref{eq-lq-cost} can be equivalently recast into a standard infinite-horizon form by introducing the damped system dynamics~\cite{perdomo2021stabilizing}. Letting $\tilde{x}_t \coloneqq \gamma^{t/2} x_t$, the closed-loop system~\eqref{eq-original system} with $u_t=-Ky_t$ is transformed into
\begin{equation}\label{eq-equal system}
    \tilde{x}_{t+1} = \sqrt{\gamma}(A - BKC)\tilde{x}_t,
\end{equation}
 and the cost function $J_\gamma(K)$ can be equivalently expressed as a standard infinite-horizon cost without the discount factor:
\begin{equation}\label{eq:standard_lqr_form}
    \begin{aligned}
        &J_{\gamma}(K)=\mathbb{E}_{\tilde{x}_0}\sum_{t=0}^{\infty}
        \tilde{x}_t^\top \left(Q+C^\top K^\top RKC\right)\tilde{x}_t,\\
        & \text{subject to}~\eqref{eq-equal system},
    \end{aligned}
\end{equation}
where $\tilde{x}_0=x_0$. 
It is well-known that the cost $J_{\gamma}(K)$ is finite if and only if $K$ belongs to the set of stabilizing policies for the damped system, defined as $\mathcal{S}_{\gamma} \coloneqq \{K \in \mathbb{R}^{m \times p} \mid \sqrt{\gamma}\rho(A-BKC) < 1 \}$~\cite{bertsekas2012dynamic}. 
According to Assumption~\ref{ass:SOF_existence}, there exists a stabilizing policy $K$ satisfying $\rho(A-BKC) < 1$. Since $\gamma \in (0, 1)$, it follows that $\sqrt{\gamma}\rho(A-BKC) \leq \rho(A-BKC) < 1$, which directly implies that the set $\mathcal{S}_{\gamma}$ is non-empty. By a slight abuse of notation, we will simply use $x_t$ to denote the state of the damped system directly in the sequel. The distinction between the original and the damped system states can be inferred from the context, given the discount factor $\gamma$.

For any $K\in\mathcal{S}_{\gamma}$, we respectively define $P_K$ and $\Sigma_{K}$ as the positive definite solution to the following Lyapunov equations:
\begin{align}
   &P_K(\gamma)=Q+C^\top K^\top RKC+\gamma\left(A-BKC\right)^\top P_K(\gamma)\left(A-BKC\right),
   \label{eq-Lya1}\\
    &\Sigma_{K}(\gamma)=I_n+\gamma\left(A-BKC\right)\Sigma_K(\gamma)\left(A-BKC\right)^\top .
    \label{eq-Lya2}
\end{align}
Then, the cost function $J_\gamma(K)$ can be succinctly expressed as
\begin{equation}\label{eq-Lya3}
J_\gamma(K) = \Tr(P_K(\gamma)) = \Tr((Q + C^\top K^\top R K C) \Sigma_K(\gamma)).
\end{equation}
Recall that our objective is to learn a stabilizing SOF policy using the PG method. To this end, in the remainder of this section, we derive several properties of the objective function $J_{\gamma}(K)$ and $K$, including local smoothness, local Lipschitz continuity, and strong stability.
To proceed, we introduce the following standard definitions.

\begin{definition}\label{defi-smoothness}(Local Smoothness)
    A function $f: \mathbb{R}^d \rightarrow \mathbb{R}$ is locally $(L,D)$-smooth over $\mathcal{X} \subseteq \mathbb{R}^d$ if for any $x\in \mathcal{X}$ and any $y\in \mathbb{R}^d$ with $\|y-x\|\leq D$,
    \begin{equation*}
        \|\nabla f(x)-\nabla f(y)\|\leq L \|x-y\|.
    \end{equation*}
\end{definition}

\begin{definition}\label{defi-Lipschitz}(Local Lipschitz Continuity)
    A function $f: \mathbb{R}^d \rightarrow \mathbb{R}$ is locally $(G, D)$-Lipschitz over $\mathcal{X} \subseteq \mathbb{R}^d$ if for any $x\in \mathcal{X}$ and any $y\in \mathbb{R}^d$ with $\|y-x\|\leq D$,
    \begin{equation*}
        |f(x)-f(y)|\leq G\|x-y\|.
    \end{equation*}
\end{definition}

\begin{definition}\label{defi-Strong stability}(Strong Stability)
    A matrix $M$ is $(\kappa,\varrho)$-strongly stable for $\kappa\geq 1$ and $0<\varrho\leq1$ if there exist matrices $H\succ0$ and $W$ such that $M=HWH^{-1}$ with $\|W\|\leq1-\varrho$ and $\|H\|\|H^{-1}\|\leq \kappa$.
\end{definition}
An SF policy $F$ is called $(\kappa,\varrho)$-strongly stable if it satisfies the following conditions:
    1) The closed-loop matrix $A-BF$ is $(\kappa,\varrho)$-strongly stable;
    2) The policy is bounded as $\|F\| \leq \kappa$.
In the sequel, we assume that 
\begin{equation}\label{eq-sys-para-norm}
    \ell_0 I \preceq Q, R \preceq \ell_1 I, \|B\| \leq \psi, \|C\| \leq \varphi,
\end{equation}
where $\ell_1\geq\ell_0>0$ and $\psi, \varphi\geq1$ are known constants.
Defining the $\nu$-sublevel set associated with a discount factor $\gamma$ as $\mathcal{S}_{\gamma}(\nu) \coloneqq \{K \mid J_{\gamma}(K) \leq \nu \}$, we first introduce the following lemma, which relates the cost of a controller to its strong stability parameters. 

\begin{lemma}\cite[Lemma 18]{cohen2019learning} 
\label{lem-strong stability}
    Suppose that an SF policy $F$ satisfies $F\in\{F~|~  J(F)\leq \nu\}$, then $F$ is $(\kappa,\varrho)$-strongly stable with $\kappa=\sqrt{\nu/\ell_0}$ and $\varrho=1/(2\kappa^2)$.
\end{lemma}

We then extend this notion to the output feedback setting. Letting $F=KC$, an SOF policy $K$ is said to be $(\kappa,\varrho)$-strongly stable if the equivalent SF gain $F$ is $(\kappa,\varrho)$-strongly stable.
For $K\in\mathcal{S}_{\gamma}(\nu)$, it follows from Lemma~\ref{lem-strong stability} that $K$ implies strong stability with $\kappa=\sqrt{\nu/\ell_0}$ and $\varrho=1/(2\kappa^2)$.
Consequently, we have
\begin{equation} \label{eq-kappa}
    \|KC\|\leq\kappa \quad \text{and} \quad \|A-BKC\|\leq \kappa(1-\varrho) < \kappa,
\end{equation}
where the bound on $\|A-BKC\|$ follows from the decomposition $HWH^{-1}$ in Definition~\ref{defi-Strong stability}.

In the following, Lemma~\ref{lem-pg-derivation} provides the expression for $\nabla J_{\gamma}(\cdot)$, while Lemmas~\ref{lem1} and~\ref{bound of gradient} establish properties of $J_{\gamma}(K)$ on the sublevel set $\mathcal{S}_{\gamma}(\nu)$; the proofs are all deferred to Appendix~\ref{app:preliminary proofs}.

\begin{lemma}(Gradient Expression) \label{lem-pg-derivation}
    For any $K \in \mathcal{S}_{\gamma}$, it holds that
    \begin{equation}
        \nabla J_{\gamma}(K) = 2 E_K \Sigma_K C^\top,
    \end{equation}
    where $E_K = (R + \gamma B^\top P_K B)KC - \gamma B^\top P_K A$, and $\Sigma_K \coloneqq \mathbb{E}_{x_0} \sum_{t=0}^{\infty} \gamma^t x_t x_t^\top$.
\end{lemma}

\begin{lemma}  \label{lem1}
    Let $K\in\mathcal{S}_{\gamma}(\nu)$ and $K'$ satisfy $\|K-K'\|\leq D$ with $D=1/(8\kappa^3\psi\varphi)$, where $\kappa=\sqrt{\nu/\ell_0}$.
    Then, we have
    \begin{enumerate}[leftmargin=*]
        \item $\|P_K\|\leq \nu$, $\|\Sigma_K\|\leq \nu/\ell_0$;
        \item $\|\Sigma_K-\Sigma_{K'}\| \leq 6\kappa^5\psi\varphi\|K-K'\|$;
        \item $\|P_K-P_{K'}\|\leq 16\kappa^7\ell_1\psi\varphi\|K-K'\|$;
        \item $J_\gamma(K)$ is locally $(G,D)$-Lipschitz over $\mathcal{S}_{\gamma}(\nu)$ with $D=1/(8\kappa^3\psi\varphi)$ and $G=16\kappa^9\ell_1\psi\varphi$;
        \item $J_{\gamma}(K)$ is locally $(L,D)$-smooth over $\mathcal{S}_{\gamma}(\nu)$ with $D=1/(8\kappa^3\psi\varphi)$ and $L=104\kappa^{10}\ell_1\psi^2\varphi^2\sqrt{\min\{m,p\}}$.
    \end{enumerate}
\end{lemma}

\begin{lemma}
   For any $K\in\mathcal{S}_{\gamma}(\nu)$, it holds that $\|\nabla J_{\gamma}(K)\|_F\leq G^0$, where $G^0=2\kappa^3\varphi(\ell_1+\psi\nu)\sqrt{\min\{m,p\}}$.
   \label{bound of gradient}
\end{lemma}

\section{Algorithm Design}
\label{section 3}
In this section, we propose the algorithm to synthesize a stabilizing SOF policy. In stark contrast with the state-feedback framework studied in~\cite{perdomo2021stabilizing,zhao2024convergence}, where the gradient dominance property permits convergence of the cost below a specific value $J$ for some $K$, i.e., $J_{\gamma}(K)\leq J$, such a property is generally absent in the output-feedback setting. Therefore, our algorithm design adopts an alternative approach by utilizing convergence to stationary points with the gradient norm condition $\|\nabla J_{\gamma}(K)\|_F \leq \epsilon$ for some $\epsilon>0$ as a criterion to synthesize a stabilizing SOF policy. However, in the model-free setting, the true gradient $\nabla J_{\gamma}(K)$ is inaccessible due to the lack of knowledge regarding the system matrices $(A, B, C)$. Therefore, we utilize a gradient estimate $\widehat{\nabla} J_{\gamma}(K)$ which is obtained by system rollouts. 
Based on the above discussions, to stabilize the system~\eqref{eq-original system}, we propose the following alternate discount method:
\begin{subequations}
\label{eq-method}
\begin{align}
&\text{(a) Find}~K_{k+1}~ \text{such that}~\|\widehat{\nabla} J_{\gamma_{k}}(K_{k+1})\|_F\leq2\epsilon/3,\label{eq-19a}\\
&\text{(b) Update}~\gamma_{k+1}=\left(1+\zeta \alpha_k\right)\gamma_k, \label{eq-19b}  
\end{align}
\end{subequations}
where $k=0,1,\ldots$ denotes the iteration index of the discount factor, $\epsilon>0$ and $\zeta \in(0,1)$ are constants, and $\alpha_k>0$ calculated via $J_{\gamma_k}(K_{k+1})$ determines the update rate of the discount factor. Note that by requiring the gradient estimation error to be bounded by $\epsilon/3$, i.e., $\|\nabla J_{\gamma_k}(K_{k+1}) - \widehat{\nabla} J_{\gamma_k}(K_{k+1})\|_F \leq \epsilon/3$, we ensure via $\|\widehat{\nabla}J_{\gamma_k}(K_{k+1})\|_F\leq2\epsilon/3$ in~\eqref{eq-19a} that $\|\nabla J_{\gamma_k}(K_{k+1})\|_F \leq \epsilon$.

For a fixed discount factor $\gamma_k$, we propose to use the zeroth-order PG method 
\begin{equation}\label{eq-pg}
    K^{j+1}=K^{j}-\eta \widehat{\nabla} J_{\gamma_k}(K^{j}),
\end{equation}
to obtain $K_{k+1}$ in \eqref{eq-19a}, starting from $K^0 = K_{k}$. 
In~\eqref{eq-pg}, $\eta >0$ is a constant stepsize and $j=0,1,\dots$ denotes the iteration index of policies. 
The iteration proceeds until a policy $K^{j}$ satisfies $\|\widehat{\nabla} J_{\gamma_{k}}(K^{j})\|_F\leq2\epsilon/3$, at which point we define $K_{k+1} = K^{j}$ to proceed to step \eqref{eq-19b}. In the remainder of this section, we will drop the index $k$ of $\gamma_k$ and denote $\gamma$ for notational simplicity when it is clear from the context.

\subsection{Gradient Estimation}
\label{subsection gradient estimation}
We use two-point methods to obtain the gradient estimate $\widehat{\nabla} J_\gamma(K)$ in Algorithm~\ref{alg:two_point_grad}, where we can simulate the system from $t=0$ to $t=\tau^e-1$ to obtain $J_{\gamma,x_0}^{\tau^e}(K)$ as follows:
\begin{align}\label{eq-value-function}
    &J_{\gamma,x_0}^{\tau^e}(K)=\sum_{t=0}^{\tau^e-1}\left(x_t^\top Qx_t+u_t^\top Ru_t\right),
\end{align}
where $u_t=-Ky_t$, and $x_0$ satisfies Assumption~\ref{ass1}.
\begin{algorithm2e}[h]
\SetNoFillComment
\caption{Two-Point Gradient Estimation}
\label{alg:two_point_grad}
\KwIn{Discount factor $\gamma$, policy $K \in \mathcal{S}_\gamma$, smoothing constant $r$, simulation time $\tau^e$, number of system trajectories $N^e$.}
\For{$i=1$ to $N^e$}{
    Sample $\overline{U}_i$ with i.i.d. entries from $\mathcal{N}(0,1)$ and normalize $U_i=\overline{U}_i/\|\overline{U}_i\|_F$.\\
    Sample an initial state $x_0^i$ from a distribution satisfying Assumption~\ref{ass1}.\\
    Simulate system~\eqref{eq-lq-cost} up to time $\tau^e$ to obtain $J_{\gamma,x_0^i}^{\tau^e}(K+rU_i)$ and $J_{\gamma,x_0^i}^{\tau^e}(K-rU_i)$.}
\KwOut{Gradient estimate $\displaystyle \widehat{\nabla} J_\gamma(K) = \frac{mp}{2rN^e} \sum_{i=1}^{N^e} \left( J_{\gamma,x_0^i}^{\tau^e}(K+rU_i) - J_{\gamma,x_0^i}^{\tau^e}(K-rU_i) \right) U_i$.}
\end{algorithm2e}

Now, we demonstrate that the gradient estimation error satisfies $\|\widehat{\nabla} J_{\gamma}(K)-\nabla J_{\gamma}(K)\|_F \leq \epsilon/3$ with high probability.

\begin{lemma}\label{lem-gradient estimation error}
   Let $\epsilon>0$. For any $\delta_0\in(0,e^{-2}/2)$, if $\tau^e$, $r$, $N^e$ in Algorithm~\ref{alg:two_point_grad} satisfy
   \begin{equation}
       \tau^e \geq \frac{2\nu}{\ell_0} \log\left(\frac{36d^2\nu^2mp}{r\epsilon\ell_0}\right),~N^e\geq\frac{162(mpG)^2(d^2+1)^2}{\epsilon^2}\log\frac{15}{\delta_0},~r\leq\min\left\{D,\frac{\nu}{G},\frac{\epsilon}{9L}\right\},
   \end{equation}
   then $ \|\widehat{\nabla}J_{\gamma}(K)-\nabla J_{\gamma}(K)\|_F\leq \epsilon/3$ holds with probability at least $1-\delta_0$, where 
   $D$, $G$, $L$ are given by Lemma~\ref{lem1}.
\end{lemma}
\textbf{Proof Sketch.} We first introduce an auxiliary function $J^r_{\gamma}(K)$:
    \begin{equation*}
        J^r_{\gamma}(K)=\mathbb{E}_{\widetilde{U}}[J_{\gamma}(K+r\widetilde{U})],~K\in\mathcal{S}_{\gamma}(\nu),
    \end{equation*}
     where the vectorized $\widetilde{U}$ is uniformly distributed from a unit ball. Accordingly, let the infinite-horizon version $\widetilde{\nabla}J_\gamma(K)$ of $\widehat{\nabla}J_\gamma(K)$ be given by
\begin{equation}
    \begin{aligned}
      \widetilde{\nabla}J_\gamma(K)=&\frac{mp}{2rN^e}\sum_{i=1}^{N^e}\left(J_{\gamma,x_0^i}(K+rU_i)-J_{\gamma,x_0^i}(K-rU_i)\right)U_i,
    \end{aligned}
    \nonumber
\end{equation}
where $J_{\gamma,x_0^i}(K)
:=\lim_{\tau^e\to\infty}
J_{\gamma,x_0^i}^{\tau^e}(K)$,
and $U_i\in\mathbb{R}^{m\times p}$ are i.i.d. random matrices uniformly distributed from the set $\{U\mid\|U\|_F=1\}$. 
The proof proceeds by decomposing the estimation error into the following three components:
\begin{equation}
\|\widehat{\nabla}J_{\gamma}(K)-\nabla J_{\gamma}(K)\|_F
        \leq \|\widehat{\nabla}J_{\gamma}(K)-\widetilde{\nabla}J_{\gamma}(K)\|_F
        +\|\widetilde{\nabla}J_{\gamma}(K)-\nabla J_{\gamma}^r(K)\|_F
        +\|\nabla J_{\gamma}^r(K)-\nabla J_{\gamma}(K)\|_F.
\nonumber
\end{equation}
In Lemma~\ref{lem5} in Appendix~\ref{proof of lem-gradient estimation error}, we establish conditions guaranteeing that $\nabla J_{\gamma}^r(K)$ and $\widetilde{\nabla}J_{\gamma}(K)$ are well-defined. We then derive upper bounds for each of the three error terms separately. More details are included in Appendix~\ref{proof of lem-gradient estimation error}. 

\subsection{Convergence Properties of PG Method}
\label{section3.2}
Based on the high-probability error bound for the gradient estimate derived in the previous subsection, we are now ready to analyze the convergence properties of the PG method in~\eqref{eq-pg} for a fixed discount factor $\gamma$. Specifically, by ensuring that the gradient estimation error $\|\widehat{\nabla} J_{\gamma}(K)-\nabla J_{\gamma}(K)\|_F$ is controlled within $\epsilon/3$, the policy $K^{j}$ obtained by~\eqref{eq-pg} is iteratively refined until the estimated gradient norm $\|\widehat{\nabla} J_{\gamma}(K)\|_F$ falls below the threshold $2\epsilon/3$. Formally, we prove the following result. 
\begin{theorem}\label{theo1}
     Suppose that $K^0\in\mathcal{S}_{\gamma}(\nu)$. For a desired accuracy $\epsilon>0$, let $\tau^e$, $N^e$, $r$, and $\eta$ satisfy 
    \begin{equation}
    \begin{aligned}
       \tau^e \geq& \frac{2\nu}{\ell_0} \log\left(\frac{36d^2\nu^2mp}{r\epsilon\ell_0}\right),~r\leq\min\Big\{D,~\frac{\nu}{G},~\frac{\epsilon}{9L}\Big\},\\
       N^e\geq&\frac{162(mpG)^2(d^2+1)^2}{\epsilon^2}\log\frac{15}{\delta_0},~ \eta\leq \min\Big\{\frac{D}{G^0+\epsilon},~\frac{1}{2 L}\Big\}.
    \end{aligned}
    \nonumber
    \end{equation}
 Then, with probability at least $1-(M+1)\delta_0$, the zeroth-order PG method in \eqref{eq-pg} based on $\widehat{\nabla}J_{\gamma}(K^{j})$ obtained by Algorithm~\ref{alg:two_point_grad} will return a policy $K^{j}$ satisfying $\|\nabla J_{\gamma}(K^{j})\|_F\leq \epsilon$ within at most
  \begin{equation}
   M\geq \frac{9J_{\gamma}(K^0)}{\eta \epsilon^2}
   \label{eq-iteration-PG}
  \end{equation}
iterations.
Here, $D$, $G$ and $L$ are given by Lemma~\ref{lem1} and $G^0$ is given by Lemma~\ref{bound of gradient}.
\end{theorem}
\begin{proof}
For notational convenience, define $w_j=\nabla J_{\gamma}(K^{j})-\widehat{\nabla}J_{\gamma}(K^{j})$ and $M=\min\{j:\|\widehat{\nabla}J_{\gamma}(K^{j})\|_F\leq 2\epsilon/3\}$. 
Suppose that $K^{j}\in\mathcal{S}_{\gamma}(\nu)$, which implies $J_{\gamma}(K^{j})\leq \nu$. It follows that
\begin{equation*}
        \|K^{j+1}-K^{j}\|_F=\eta\|\widehat{\nabla}J_{\gamma}(K^{j})\|_F
                         \leq \eta\left(\|\nabla J_{\gamma}(K^{j})\|_F+\|w_j\|_F\right)
                         \overset{(a)}{\leq}\eta(G^0+\epsilon/3)\overset{(b)}{\leq}D,
\end{equation*}
where (a) follows from Lemma~\ref{bound of gradient} based on $K^{j}\in\mathcal{S}_{\gamma}(\nu)$, and Lemma~\ref{lem-gradient estimation error}, which guarantees $\|w_j\|_F \leq \epsilon/3$ under the chosen parameters $\tau^e$, $r$, and $N^e$; (b) holds by the choice of $\eta$. Thus, we conclude that $K^{j+1}$, $K^{j}$ satisfy the conditions of local smoothness described in Definition~\ref{defi-smoothness}, which further implies that
\begin{equation}\label{eq-29}
    \begin{aligned}
            J_{\gamma}(K^{j+1})-J_{\gamma}(K^{j})
            &\leq \langle\nabla J_{\gamma}(K^{j}),K^{j+1}-K^{j}\rangle+\frac{L}{2}\|K^{j+1}-K^{j}\|^2_F\\
            &=-\eta\langle\nabla J_{\gamma}(K^{j}),\widehat{\nabla} J_{\gamma}(K^{j})\rangle+\frac{L\eta^2}{2}\|\widehat{\nabla} J_{\gamma}(K^{j})\|_F^2\\
            &=-\eta\left[\|\widehat{\nabla} J_{\gamma}(K^{j})\|_F^2+\langle w_j,\widehat{\nabla} J_{\gamma}(K^{j}) \rangle-\frac{L\eta}{2}\|\widehat{\nabla} J_{\gamma}(K^{j})\|_F^2\right]\\
            &\leq-\eta\left[\left\|\frac{1}{2}\widehat{\nabla}J_{\gamma}(K^{j})+w_j\right\|_F^2+\left(\frac{3}{4}-\frac{L\eta}{2}\right)\|\widehat{\nabla}J_{\gamma}(K^{j})\|_F^2-\|w_j\|_F^2\right]\\
            &\leq-\eta\left(\frac{3-2L\eta}{4}\|\widehat{\nabla} J_{\gamma}(K^{j})\|_F^2-\|w_j\|_F^2\right)\\
            &\overset{(a)}{\leq}-\eta\left(\frac{1}{2}\|\widehat{\nabla} J_{\gamma}(K^{j})\|_F^2-\frac{\epsilon^2}{9}\right),\\
    \end{aligned}
\end{equation}
where (a) follows from the choice of $\eta$ such that $L\eta\leq1/2$, and $\|w_j\|_F\leq \epsilon/3$ which is obtained by applying Lemma~\ref{lem-gradient estimation error} based on the choice of $\tau^e$, $r$ and $N^e$.
For $0\leq j\leq M-1$, the definition of $M$ implies $ \|\widehat{\nabla}J_{\gamma}(K^{j})\|_F > 2\epsilon/3$, which, upon substitution into~\eqref{eq-29}, yields
\begin{equation}\label{eq-29-1}
     J_{\gamma}(K^{j+1})- J_{\gamma}(K^{j})\leq-\frac{\eta \epsilon^2}{9}< 0.
\end{equation}
Consequently, we get $K^{j+1}\in\mathcal{S}_{\gamma}(\nu)$. 

Now, given that $K^0\in\mathcal{S}_{\gamma}(\nu)$, the iterate $K^1$ generated by~\eqref{eq-pg} remains within $\mathcal{S}_{\gamma}(\nu)$. By mathematical induction, we can establish that for all $j=0,\ldots, M-1$, \eqref{eq-29-1} holds.
Summing~\eqref{eq-29-1} for $j=0,\ldots,M-1$, and utilizing the fact that $\|\widehat{\nabla} J_{\gamma}(K^{j})\|_F > 2\epsilon/3$ for all $j < M$, we obtain
\begin{equation}
    \begin{aligned}
        J_{\gamma}(K^{M}) 
        \leq J_{\gamma}(K^{0}) - \sum_{j=0}^{M-1} \frac{\eta\epsilon^2}{9} 
        = J_{\gamma}(K^{0}) - \frac{M\eta\epsilon^2}{9}.
    \end{aligned}
    \nonumber
\end{equation}
Rearranging the terms and noting that $J_{\gamma}(K^{M}) \geq 0$, we have
\begin{equation}
    \frac{M\eta\epsilon^2}{9} \leq J_{\gamma}(K^{0}) - J_{\gamma}(K^{M}) \leq J_{\gamma}(K^{0}).
    \nonumber
\end{equation}
Hence, the condition $\|\widehat{\nabla}J_{\gamma}(K^M)\|_F \leq 2\epsilon/3$ is guaranteed to be satisfied within at most
\begin{equation}
    \frac{9 J_{\gamma}(K^0)}{\eta \epsilon^2}
    \nonumber
\end{equation}
iterations.

Moreover, recall that under the choice of $\tau^e$, $r$ and $N^e$ in Theorem~\ref{theo1}, $\|\widehat{\nabla}J_{\gamma}(K^{j})-\nabla J_{\gamma}(K^{j})\|_F\leq \epsilon/3$ holds for any single iteration with probability at least $1-\delta_0$ according to Lemma~\ref{lem-gradient estimation error}. 
Taking the union bound for $j=0,\ldots,M$, we obtain that $\|\widehat{\nabla}J_{\gamma}(K^M)-\nabla J_{\gamma}(K^M)\|_F\leq \epsilon/3$ holds with probability at least $1-(M+1)\delta_0$. Combining $\|\nabla J_{\gamma}(K^M)\|_F\leq\|\widehat{\nabla}J_{\gamma}(K^M)\|_F +\|\widehat{\nabla}J_{\gamma}(K^M)-\nabla J_{\gamma}(K^M)\|_F$, we conclude that with probability at least $1-(M+1)\delta_0$, $\|\nabla J_{\gamma}(K^M)\|_F\leq \epsilon$ holds. 
Since $K^0$ is a stabilizing policy and $J_{\gamma}(K^{j})$ is strictly decreasing for $j=0,\ldots,M-1$, $K^M$ generated by the zeroth-order PG method~\eqref{eq-pg}  is also a stabilizing policy with $\|\nabla J_{\gamma}(K^M)\|_F\leq \epsilon$. 
The proof is now completed.
\end{proof}

{\bf Discussions.} As mentioned at the beginning of Section~\ref{section 3}, Theorem~\ref{theo1} guarantees convergence to an $\epsilon$-stationary point satisfying $\|\nabla J_{\gamma}(K^{j})\|_F \leq \epsilon$, which is different from the results in~\cite{perdomo2021stabilizing,zhao2024convergence} that allow convergence of the cost function to a specific cost level. Such a difference is caused by the facts that the set of stabilizing SOF controllers may be disconnected and the cost landscape is riddled with saddle points~\cite{fatkhullin2021optimizing}. 
Accordingly, our zeroth-order gradient estimate plays two roles: (i) constructing the policy update in~\eqref{eq-pg}; (ii) defining the implementable stopping condition $\|\widehat{\nabla}J_{\gamma}(K^j)\|_F\leq 2\epsilon/3$ of the PG iterations. 
We establish that for a fixed discount factor $\gamma$, the proposed zeroth-order PG method can reach a region of SOFs with small gradient norm while the discounted cost strictly decreases along the way (i.e., $J_{\gamma}(K^{j+1}) < J_{\gamma}(K^j)$). This monotonic decrease, in turn, facilitates a progressive increase of the discount factor $\gamma$.

\subsection{Learning to Stabilize Systems by Discount Methods}
\label{Learning to Stabilize Systems by Discount Methods}
Building upon the convergence of the PG method in Section~\ref{section3.2}, we now elaborate on the evolution of the discount factor in~\eqref{eq-19b}. We propose an adaptive updating rule $\gamma_{k+1}=(1+\zeta \alpha_k)\gamma_k$, which is central to navigating the policy from the initial discounted regime toward the stability region $\mathcal{S}_1$ of the original system~\eqref{eq-original system}. To conclude the design phase, we present an algorithmic framework for learning a stabilizing SOF policy through the discount method.

For a policy $K\in\mathcal{S}_{\gamma}$, we now show how to find a larger discount factor $\gamma'>\gamma$ such that $K$ remains in the stabilizing set $\mathcal{S}_{\gamma'}$. 

\begin{lemma}\label{lem8}
    Let $K\in\mathcal{S}_\gamma$. If $\gamma'<(1+\alpha)\gamma$ with
    \begin{equation}\label{eq-updating a}
        \alpha=\frac{\ell_0}{J_\gamma(K)-\ell_0},
    \end{equation}
    then $K\in\mathcal{S}_{\gamma'}$.
\end{lemma}

The condition derived in Lemma~\ref{lem8}, whose proof is deferred to Appendix~\ref{proof of lem8}, provides a theoretical basis for updating the discount factor. 
However, the calculation of the scalar $\alpha$ in \eqref{eq-updating a} requires the exact value of the cost $J_\gamma(K)$, which is typically unknown in a model-free setting. 
Note that $\ell_0$ is a known constant satisfying $\ell_0 I\preceq Q$. 
Therefore, to implement the update rule, we replace $J_\gamma(K)$ with an estimate obtained from system trajectories. 
Define the empirical cost estimator $\widehat{J}_{\gamma}^{\tau,N}(K)$ as the average of the truncated costs over $N$ independent trajectories:
\begin{equation}\label{eq-estimator-def}
\widehat{J}_{\gamma}^{\tau,N}(K) = \frac{1}{N} \sum_{i=1}^N J_{\gamma,x_0^i}^{\tau}(K),
\end{equation}
where $J_{\gamma,x_0^i}^{\tau}(K)$ is defined in~\eqref{eq-value-function} and $N\in \mathbb{N}_+$. The following lemma characterizes the estimation accuracy of~\eqref{eq-estimator-def}.

\begin{lemma}~\cite[Lemma 5]{zhao2024convergence}
\label{lem4}
    Let $K\in \mathcal{S}_\gamma$ and $0<\delta_1<1$ be a constant. If the simulation time $\tau$ and the number of system trajectories $N$ satisfy
    \begin{equation}\label{eq-tau-N}
        \tau\geq \frac{2J_\gamma(K)}{\ell_0}\log\left(\frac{J_\gamma(K)d^2}{\ell_0}\right),
        N\geq 8d^4\log\frac{2}{\delta_1},
    \end{equation}
    we obtain that the following holds with probability at least $1-\delta_1$:
    \begin{equation}\label{eq-iteration-PG-hatJ}
        |J_\gamma(K)-\widehat{J}_\gamma^{\tau,N}(K)|\leq \frac{1}{2}J_\gamma(K),
    \end{equation}
    where $\widehat{J}_\gamma^{\tau,N}(K)$ is given by~\eqref{eq-estimator-def}.
\end{lemma}

Lemma~\ref{lem4} shows that the cost $J_{\gamma}(K)$ can be estimated by $\widehat{J}_{\gamma}^{\tau,N}(K)$ with high probability. Specifically, the inequality~\eqref{eq-iteration-PG-hatJ} implies that $J_{\gamma}(K) \leq 2\widehat{J}_{\gamma}^{\tau,N}(K)$, and we have $\ell_0/(2\widehat{J}_{\gamma}^{\tau,N}(K)-\ell_0)\leq \ell_0/(J_\gamma(K)-\ell_0)$. 
Therefore, to design a practical update rule for the discount factor $\gamma$, we substitute the cost $J_{\gamma}(K)$ with its upper bound $2\widehat{J}_{\gamma}^{\tau, N}(K)$. Furthermore, we introduce a scaling parameter $\zeta \in (0, 1)$ to ensure the condition $\gamma'<(1+\alpha)\gamma$ as discussed in Lemma~\ref{lem8}. 
Combining these, the update rate for $\gamma_k$ is designed as $\gamma_{k+1}=(1+\zeta \alpha_k)\gamma_k$ with
\begin{equation}\label{eq-updating a-new}
    \alpha_k = \frac{\ell_0}{2\widehat{J}_{\gamma_k}^{\tau_k,N}(K_{k+1}) - \ell_0},
\end{equation}
where $K_{k+1}\in\mathcal{S}_{\gamma_k}$ is obtained by the PG method in~\eqref{eq-pg}, which satisfies $\|\widehat{\nabla} J_{\gamma_{k}}(K_{k+1})\|_F\leq 2\epsilon/3$.

Based on the update rate for the discount factor $\gamma_k$, we present our algorithmic framework for learning a stabilizing SOF policy in Algorithm~\ref{alg:learning_stabilize}. Note that for any $k=0,1,\dots$, the gradient estimate $\widehat{\nabla}J_{\gamma_k}(K^j)$ in lines~3-4 of Algorithm~\ref{alg:learning_stabilize} is obtained by Algorithm~\ref{alg:two_point_grad}, where the discount factor $\gamma_k$ is used and the parameters $\tau^e$, $N^e$ and $r$ of Algorithm~\ref{alg:two_point_grad} also become dependent on the index $k$ in Algorithm~\ref{alg:two_point_grad}.

\begin{algorithm2e}[h]
\SetNoFillComment
\caption{Learning SOF}
\label{alg:learning_stabilize}
\KwIn{Initial policy $K_{0}=0$, discount factor $\gamma_0<1/\rho^2(A)$, desired accuracy $\epsilon>0$, decay factor $0<\zeta<1$, simulation time $\tau_k$, number of system trajectories $N$, stepsize $\eta_k$.}
\For{$k=0,1,\dots$}{
    Let $j=0$ and $K^0=K_{k}$.\\
    \While{$\|\widehat{\nabla}J_{\gamma_k}(K^{j})\|_F > 2\epsilon/3$}{
        Update $K^{j+1} = K^{j} - \eta_k \widehat{\nabla}J_{\gamma_k}(K^{j})$.\\
        Set $j = j+1$.\\
    }
    Set $K_{k+1}=K^{j}$.\\
    Update $\gamma_{k+1}=(1+\zeta \alpha_k)\gamma_k$ with $\alpha_{k}=\frac{\ell_0}{2\widehat{J}_{\gamma_k}^{\tau_k,N}(K_{k+1})-\ell_0}$.\\
    \If{$\gamma_{k+1}\geq 1$}{
        \Return $K_{k+1}$.
    }
}
\KwOut{policy $K_{k+1}$.}
\end{algorithm2e}

\section{Algorithm Analysis}
We now analyze the convergence and sample complexity of Algorithm~\ref{alg:learning_stabilize}, which are the main results of this work. We begin by identifying the conditions in Section~\ref{Convergence Analysis} under which  Algorithm~\ref{alg:learning_stabilize} can stabilize the original system~\eqref{eq-original system} with high probability. In Section~\ref{Sample Complexity}, we characterize the overall sample complexity of learning a stabilizing SOF policy via the discount method.

\subsection{Convergence Analysis}
\label{Convergence Analysis}
The convergence analysis of Algorithm~\ref{alg:learning_stabilize} will be built upon the analysis of gradient estimation error in Section~\ref{subsection gradient estimation}, the convergence properties of PG methods in Section~\ref{section3.2}, and the update strategy for the discount factor $\gamma_k$ in Section~\ref{Learning to Stabilize Systems by Discount Methods}. 
Unlike the state-feedback analysis in~\cite{perdomo2021stabilizing,zhao2024convergence}, which uses gradient dominance to reach a prescribed
cost level, our SOF analysis combines convergence to $\epsilon$-stationary
policies with a uniform cost bound to establish finite-time stabilization.
Specifically, the following theorem specifies sufficient conditions on various parameters in Algorithm~\ref{alg:two_point_grad}-\ref{alg:learning_stabilize} so that Algorithm~\ref{alg:learning_stabilize} returns a stabilizing SOF policy with high probability.

\begin{theorem}\label{theorem2}
     Let constants $0<\delta_0<e^{-2}/2$ and $0<\delta_1<1$, and consider any $\epsilon>0$. For any $k\geq0$,
     let the parameters $\tau_k,N,\eta_k$ of Algorithm~\ref{alg:learning_stabilize} satisfy
     \begin{equation}
    \begin{aligned}\label{eq-theo2-1}
              \tau_k\geq\frac{2\nu_k}{\ell_0}\log\frac{\nu_kd^2}{\ell_0},
              N\geq 8d^4\log\frac{2}{\delta_1},
              \eta_k= \min\Big\{\frac {D_k}{G_k^0+\epsilon},\frac{1}{2L_k}\Big\},
         \end{aligned}
     \end{equation}
     and let the parameters $\tau_k^e,N_k^e,r_k$ of Algorithm~\ref{alg:two_point_grad} satisfy
     \begin{equation}
         \begin{aligned} \label{eq-theo2-2}
           \tau_k^e \geq \frac{2\nu_k}{\ell_0} \log\left(\frac{36d^2\nu_k^2mp}{r_k\epsilon\ell_0}\right),
           N_k^e\geq\frac{162(mpG_k)^2(d^2+1)^2}{\epsilon^2}\log\frac{15}{\delta_0},
           r_k\leq\min\Big\{D_k,~\frac{\nu_k}{G_k},~\frac{\epsilon}{9L_k}\Big\},
         \end{aligned}
     \end{equation}   
     where $\nu_0=J_{\gamma_0}(K_0)$ and $\nu_k={8\big(\widehat{J}_{\gamma_{k-1}}^{\tau_{k-1},N}(K_{k})\big)^3}/((1-\zeta)\ell_0^2)$ if $k\ge1$ with $\widehat{J}_{\gamma_{k-1}}^{\tau_{k-1},N}(K_{k})$ given by line~7 of Algorithm~\ref{alg:learning_stabilize}; 
     $D_k$, $L_k$, $G_k$ and $G^0_k$ are given by Lemmas~\ref{lem1}-\ref{bound of gradient} with $\nu$ replaced by $\nu_k$. 
     If the iteration $k$ in Algorithm~\ref{alg:learning_stabilize} exceeds $k'$ with
    \begin{equation*}
        k'=\frac{\log(1/\gamma_0)}{\log(1+\zeta\ell_0/(3\overline{J}-\ell_0))},
    \end{equation*}
where $\overline{J}:=\sup_{(\gamma,K)\in\mathcal K_\epsilon}J_\gamma(K)$ and  $\mathcal{K}_\epsilon:=\left\{(\gamma,K): \gamma\in[\gamma_0,1], K\in\mathcal{S}_\gamma, \|\nabla J_\gamma(K)\|_F\leq\epsilon\right\}$,
Algorithm~\ref{alg:learning_stabilize} returns a stabilizing policy $K_{k'+1}$ 
    with probability at least $1-k'\delta_1-\sum_{k=0}^{k'}(M_k+1)\delta_0$, where 
    $M_k=9\nu_k/(\eta_k\epsilon^2)$.
\end{theorem}

{\bf Required information for parameter selection.} According to Theorem~\ref{theorem2}, 
the parameters $\tau_k,\eta_k,\tau_k^e,N_k^e,r_k$ of Algorithms~\ref{alg:two_point_grad}-\ref{alg:learning_stabilize} are set at the beginning of each iteration $k$ of the for~loop from lines~1-9 in Algorithm~\ref{alg:learning_stabilize}, using $\nu_k$ in Lemmas~\ref{lem1}-\ref{bound of gradient}, where $\nu_k$ is computed based on $\widehat{J}_{\gamma_{k-1}}^{\tau_{k-1},N}(K_{k})$ obtained from the last iteration (i.e., $(k-1)$-th iteration). Note that for $k=0$, we let $\nu_0 = J_{\gamma_0}(0)$ which is assumed to be known. Hence, we see from the above arguments (and also Lemmas~\ref{lem1}-\ref{bound of gradient}) that to set the parameters required by Algorithms~\ref{alg:two_point_grad}-\ref{alg:learning_stabilize}, we only need the knowledge of the norm bounds $\psi,\varphi$ on the unknown system matrices $B$ and $C$ specified in \eqref{eq-sys-para-norm}, and the value of $J_{\gamma_0}(0)$ which can also be estimated using system rollouts from \eqref{eq-lq-cost}. Since $\nu_k\leq\overline{\nu}$ uniformly as shown below, all
parameters can be selected using $\overline{\nu}$.

{\bf The value of $\overline{J}$.}
We show that the uniform cost bound $\overline{J}$ in
Theorem~\ref{theorem2} is finite under a verifiable assumption. We first define $\mathcal{D}
:=\left\{(\gamma,F): \gamma\in[\gamma_0,1], F=KC, K\in\mathcal{S}_{\gamma}\right\}$ and $\mathcal{E}_{\epsilon}:=\left\{(\gamma,KC):(\gamma,K)\in\mathcal{K}_{\epsilon}\right\}$.
Assume that $\overline{\mathcal{E}_{\epsilon}}\subseteq
\mathcal{D}$ and $\overline{\mathcal{E}_{\epsilon}}$ is compact, where $\overline{\mathcal{E}_{\epsilon}}$ denotes the closure of
$\mathcal{E}_{\epsilon}$.
Since the closed-loop dynamics and the cost depend on $K$ only through $F=KC$, define
$\widetilde{J}_{\gamma}(F):=J_{\gamma}(K)$ for $F=KC$. Then, the uniform cost bound is 
$\overline{J}:=\sup_{(\gamma,K)\in\mathcal{K}_{\epsilon}}J_{\gamma}(K)
=\sup_{(\gamma,F)\in\mathcal{E}_{\epsilon}}
\widetilde{J}_{\gamma}(F)$.
For every $(\gamma,F)\in\mathcal{D}$, let $A_F:=A-BF$ and
$Q_F:=Q+F^\top RF$. Since $\sqrt{\gamma}\rho(A_F)<1$, the matrix
$I-\gamma(A_F^\top\otimes A_F^\top)$ is nonsingular, and the corresponding discounted Lyapunov equation as per~\eqref{eq-Lya1} (where we let $KC=F$) has the unique solution
$\operatorname{vec}(P_{\gamma,F})
=[I-\gamma(A_F^\top\otimes A_F^\top)]^{-1}\operatorname{vec}(Q_F)$.
Therefore, $P_{\gamma,F}$ is a continuous function of $(\gamma,F)$ over $\mathcal{D}$, which implies that $\widetilde J_{\gamma}(F)=\Tr(P_{\gamma,F})$ is also a continuous function of $(\gamma,F)$ over $\mathcal{D}$.
Since $\overline{\mathcal{E}_{\epsilon}}$ is a compact subset of $\mathcal{D}$, the extreme-value theorem~\cite{rudin2021principles} yields $\overline{J}\leq\max_{(\gamma,F)\in\overline{\mathcal{E}_{\epsilon}}}\widetilde{J}_{\gamma}(F)<\infty$.
Additionally, for single-input multi-output systems, one can verify that the above assumption holds for any $\epsilon>0$ if $(A,B)$ is controllable, $A=A^\top$, $C=aB^\top$ for $a\in\mathbb{R}^{p}\setminus\{\mathbf{0}\}$, $A$ has only one eigenvalue outside the unit disk, and $B^\top(I-A^2)^{-1}B<0$; see Proposition~\ref{prop:symmetric-collocated-SIMO} in Appendix~\ref{sec:SIMO} for more details and a potential extension to multi-input multi-output systems. Finally, the numerical example in Section~\ref{sec:numerical_example} satisfies the sufficient structural conditions in Proposition~\ref{prop:symmetric-collocated-SIMO}. Although the cart-pole example in Section~\ref{sec:application_cart_pole} does not satisfy the above conditions, Algorithm~\ref{alg:learning_stabilize} still returns a stabilizing SOF controller successfully.

{\bf High probability convergence.} The success probability characterized in Theorem~\ref{theorem2} depends on $M_k$, which in turn depends on $\eta_k$ and $\nu_k$ computed during the execution of the algorithm (as we discussed above), for all $k=0,\dots,k'$. We now argue a uniform bound on $M_k$ over all $k=0,\dots,k'$, which will yield a success probability bound in terms of $J_{\gamma_0}(K_0)$ and $\overline{J}$.
Note that in Theorem~\ref{theorem2}, we set $\nu_k=8\big(\widehat{J}_{\gamma_{k-1}}^{\tau_{k-1},N}(K_{k})\big)^3/((1-\zeta)\ell_0^2)$ for all $k\geq1$. By Lemma~\ref{lem12} that will be shown in the proof of Theorem~\ref{theorem2} below, we obtain $\widehat{J}_{\gamma_{k-1}}^{\tau_{k-1},N}(K_{k})\leq 3\overline{J}/2$ with high probability, which yields that $\nu_k\leq \widetilde{\nu}=27\overline{J}^3/((1-\zeta)\ell_0^2)$ for $k\geq1$. Combined with $\nu_0=J_{\gamma_0}(K_0)$, we have $\nu_k\leq \overline{\nu}=\max\{J_{\gamma_0}(K_0), \widetilde{\nu}\}$ for all $k\geq0$.
It follows from Lemma~\ref{lem1} that there exist $\underline{D}$, $\overline{G}^0$ and $\overline{L}$ such that $\underline{D}\leq D_k$, $\overline{G}^0\geq G^0_k$ and $\overline{L}\geq L_k$ for all $k=0,\dots, k'$. Now, we have $\eta_k\geq \underline{\eta}=\min\{\underline{D}/(\overline{G}^0+\epsilon),1/(2\overline{L})\}$ holds for all $k=0,\dots, k'$. Consequently, the upper bound
\begin{equation}\label{eq-add1}
    M_k\leq \overline{M}=9\overline{\nu}/(\underline{\eta}\epsilon^2)
\end{equation}
holds for all $k=0,\dots, k'$. As a result, Algorithm~\ref{alg:learning_stabilize} returns a stabilizing policy $K_{k'+1}$ with probability at least $1-k'\delta_1-(k'+1)(\overline{M}+1)\delta_0$.

{\bf Proof of Theorem~\ref{theorem2}.} The proof relies on the following three key technical lemmas proved in Appendix~\ref{Proofs Omitted in Theorem 2}.
\begin{lemma} \label{lem6}
    Suppose that $J_{\gamma_{k}}(K_{k+1})\leq \overline{J}$ for $\overline{J}>0$. Let $\zeta\in(0,1)$ be a decay factor and $\gamma_{k+1}=(1+\zeta \alpha_k)\gamma_k$, where $\alpha_k$ is given in~\eqref{eq-updating a-new}. If $\tau_k$ and $N$ satisfy the conditions of Lemma~\ref{lem4}, then $\sqrt{\gamma_{k+1}}\rho(A-BK_{k+1}C)<1$.
\end{lemma}

\begin{lemma} \label{lem7}
      Suppose that $J_{\gamma_{k}}(K_{k+1})\leq \overline{J}$ for $\overline{J}>0$ and $|J_{\gamma_k}(K_{k+1})-\widehat{J}_{\gamma_k}^{\tau_k,N}(K_{k+1})|\leq J_{\gamma_k}(K_{k+1})/2$. Then,
        \begin{equation}
            \alpha_ k\geq \frac{\ell_0}{3\overline{J}-\ell_0},~
            J_{\gamma_{k+1}}(K_{k+1})\leq \frac{8\left(\widehat{J}_{\gamma_k}^{\tau_k,N}(K_{k+1})\right)^3}{(1-\zeta)\ell_0^2}.
        \end{equation}
\end{lemma}

\begin{lemma}\label{lem12}
    Let $\overline{J}>0$. Under the parameters in Theorem~\ref{theorem2}, if at iteration $k$ the following conditions hold
    \begin{subequations}
    \label{eq-lem12}
      \begin{align}
        &J_{\gamma_{k}}(K_{k+1})\leq \overline{J},\label{eq-lem12-1}\\
        &|J_{\gamma_{k}}(K_{k+1})-\widehat{J}_{\gamma_{k}}^{\tau_{k},N}(K_{k+1})|\leq J_{\gamma_{k}}(K_{k+1})/2,\label{eq-lem12-2}
      \end{align}
    \end{subequations}
    then they still hold at iteration $k+1$ with probability at least $1-\delta_1-(M_{k+1}+1)\delta_0$, where $M_{k+1}$ is given in Theorem~\ref{theorem2}.
\end{lemma}

Based on the lemmas above, we are now ready to prove Theorem~\ref{theorem2}.
The proof proceeds by induction. The overall idea is to show that if the condition $\sqrt{\gamma_{k}}\rho(A-BK_{k}C)<1$ holds at iteration $k$, then it continues to hold at iteration $k+1$ with high probability. Moreover, a stabilizing policy $K_{k+1}$ for the system \eqref{eq-original system} is successfully obtained once the discount factor $\gamma_{k+1}$ satisfies the criterion $\gamma_{k+1}\geq1$. 
Under the assumption that $\overline{\mathcal{E}_{\epsilon}}\subseteq
\mathcal{D}$ and $\overline{\mathcal{E}_{\epsilon}}$ is compact, 
it follows from Theorem~\ref{theo1} that for each $\gamma_k\in[\gamma_0,1]$, the zeroth-order PG method in~\eqref{eq-pg} returns a stabilizing policy satisfying $J_{\gamma_k}(K_{k+1})\leq \overline{J}$.

  \textbf{Base Case:} At iteration $0$, $K_{0}=0$ only needs to stabilize the initial damped system. Since $\gamma_0<1/\rho^2(A)$ as specified in Algorithm~\ref{alg:learning_stabilize}, we have $\sqrt{\gamma_0}\rho(A-BK_{0}C)=\sqrt{\gamma_0}\rho(A)<1$, and hence $K_0\in\mathcal{S}_{\gamma_0}$.
   We next show that
    \begin{subequations}
       \label{eq-theo2-base}
      \begin{align}
        &J_{\gamma_{0}}(K_{1})\leq \overline{J}, \label{eq-theo2-base1}\\
        &|J_{\gamma_{0}}(K_{1})-\widehat{J}_{\gamma_{0}}^{\tau_0,N}(K_{1})|\leq  J_{\gamma_{0}}(K_{1})/2, \label{eq-theo2-base2}
       \end{align}
    \end{subequations}
    hold with probability at least $1-\delta_1-(M_0+1)\delta_0$, while ensuring $\sqrt{\gamma_1}\rho(A-BK_1C)<1$. 
 
  Note that the parameters $\tau_0^e$, $N_0^e$, $r_0$ and $\eta_0$ chosen in Theorem~\ref{theorem2} also satisfy the conditions required by Theorem~\ref{theo1}, and $M_0$ satisfies the sufficient condition in~\eqref{eq-iteration-PG} if we choose 
  \begin{equation*}
      M_0=\frac{9\nu_0}{\eta_0\epsilon^2}=\frac{9J_{\gamma_0}(0)}{\eta_0\epsilon^2}.
  \end{equation*}
  Thus, by Theorem~\ref{theo1}, with probability at least $1-(M_0+1)\delta_0$, the zeroth-order PG method in~\eqref{eq-pg} will return a stabilizing policy $K_1$ satisfying $\|\nabla J_{\gamma_0}(K_1)\|_F\leq \epsilon$ within at most $M_0$ iterations. From the definition of $\overline{J}$, we can conclude that~\eqref{eq-theo2-base1} holds with probability at least $1-(M_0+1)\delta_0$.
  Moreover, since $J_{\gamma_0}(K_1)\leq J_{\gamma_0}(K_0)=\nu_0$, we have
  \begin{equation}
      \tau_0\geq \frac{2\nu_0}{\ell_0}\log\frac{\nu_0d^2}{\ell_0}
      \geq \frac{2J_{\gamma_0}(K_1)}{\ell_0}\log\frac{J_{\gamma_0}(K_1)d^2}{\ell_0}.
      \nonumber
  \end{equation}
  The above inequality combined with the choice of $N$ in Theorem~\ref{theorem2} satisfy the conditions in~\eqref{eq-tau-N} of Lemma~\ref{lem4}.
  Therefore, by Lemma~\ref{lem4}, we obtain that~\eqref{eq-theo2-base2} holds with probability at least $1-\delta_1$.
  Applying a union bound, we deduce that~\eqref{eq-theo2-base} holds with probability at least $1-\delta_1-(M_0+1)\delta_0$. 
  Finally, we can invoke Lemma~\ref{lem6} based on~\eqref{eq-theo2-base} to show that $\sqrt{\gamma_{1}}\rho(A-BK_{1}C)<1$.

    \textbf{Inductive Step:} Suppose that at iteration $k$, the following conditions hold,
    \begin{subequations}
        \label{eq-proof-theo2-1}
        \begin{align}
            &J_{\gamma_{k}}(K_{k+1})\leq \overline{J},\\
            &|J_{\gamma_{k}}(K_{k+1})-\widehat{J}_{\gamma_{k}}^{\tau_{k},N}(K_{k+1})|\leq J_{\gamma_{k}}(K_{k+1})/2.
        \end{align}
    \end{subequations}
    Due to $\overline{J}_{\gamma_{k}}\leq \overline{J}$ for all $\gamma_{k}\in(0,1]$, it follows from Lemma~\ref{lem12} that the conditions in~\eqref{eq-proof-theo2-1} still hold at iteration $k+1$ with probability at least $1-\delta_1-(M_{k+1}+1)\delta_0$.
    Furthermore, given that the conditions in~\eqref{eq-proof-theo2-1} are satisfied, invoking Lemma~\ref{lem6} guarantees that $\sqrt{\gamma_{k+1}}\rho(A-BK_{k+1}C)<1$ holds.
    Applying induction, we conclude that if $\sqrt{\gamma_{k}}\rho(A-BK_{k}C)<1$ holds, then it still holds at iteration $k+1$ with probability at least $1-\delta_1-(M_{k+1}+1)\delta_0$.
    
    Moreover, the validity of~\eqref{eq-proof-theo2-1} allows us to apply Lemma~\ref{lem7}, which implies $\alpha_k\geq\underline{\alpha}$ with $\underline{\alpha}=\ell_0/(3\overline{J}-\ell_0)$; this serves as a uniform lower bound for all $k$. Thus, 
    \begin{equation}
    \gamma_k = \gamma_0 \prod_{j=0}^{k-1} (1 + \zeta \alpha_j) \geq  \gamma_0 \prod_{j=0}^{k-1} (1 + \zeta \underline{\alpha}) = \gamma_0 (1 + \zeta \underline{\alpha})^k.
    \nonumber
    \end{equation}
    Letting $\gamma_0 (1 + \zeta \underline{\alpha})^k\geq1$, we obtain that with probability at least $1-k'\delta_1-\sum_{k=0}^{k'}(M_k+1)\delta_0$, Algorithm~\ref{alg:learning_stabilize} will return a stabilizing policy $K_{k'+1}$ provided that the number of iterations satisfies $k\geq k'$ where
    \begin{equation*}
        k'=\frac{\log(1/\gamma_0)}{\log(1+\zeta\underline{\alpha})}
        =\frac{\log(1/\gamma_0)}{\log(1+\zeta\ell_0/(3\overline{J}-\ell_0))}.
    \end{equation*}
    The proof is completed.    $\hfill\blacksquare$

\subsection{Sample Complexity}
\label{Sample Complexity}
Theorem \ref{theorem2} shows that Algorithm~\ref{alg:learning_stabilize} in tandem with Algorithm~\ref{alg:two_point_grad} converges to a stabilizing SOF policy for an unstable system within a finite number of iterations.
We now explicitly characterize the sample complexity of Algorithm~\ref{alg:learning_stabilize}, denoted by $\mathcal{N}_{\text{total}}$, which represents the total number of {\it system trajectories} required to achieve the result in Theorem~\ref{theorem2}.

The system trajectories are generated in two phases: the gradient estimation in Algorithm~\ref{alg:two_point_grad} and the cost evaluation for the discount factor update in Algorithm~\ref{alg:learning_stabilize}. As we argued in Section~\ref{Convergence Analysis}, since $\nu_k\leq \overline{\nu}$ holds for all $k=0,\ldots,k'$, there 
exists uniform bound $\overline{G}$ such that $G_k\leq\overline{G}$ for all $k=0,\ldots,k'$. 
Under the parameter settings of Algorithm~\ref{alg:two_point_grad} and Algorithm~\ref{alg:learning_stabilize}, we obtain that the total number of trajectories is
\begin{equation}
  \begin{aligned}
     \mathcal{N}_{\text{total}}
     &=2\sum_{k=0}^{k'}M_kN_k^e+k'N\\
     &\overset{(a)}{\leq}\frac{\log(1/\gamma_0)}{\log(1+\zeta\ell_0/(3\overline{J}-\ell_0))}\left(  \frac{324\overline{M}(mp\overline{G})^2(d^2+1)^2}{\epsilon^2}\log\frac{15}{\delta_0} +8d^4\log\frac{2}{\delta_1}\right)\\
     &\overset{(b)}{\leq}\frac{3\overline{J}\log(2\rho(A)^2)}{\zeta\ell_0}\left(  \frac{324\overline{M}(mp\overline{G})^2(d^2+1)^2}{\epsilon^2}\log\frac{15}{\delta_0}+8d^4\log\frac{2}{\delta_1}\right).
  \end{aligned}
  \nonumber
\end{equation}
For (a), we apply $M_k\leq \overline{M}$ in~\eqref{eq-add1} for all $k=0,\ldots,k'$.
For (b), we use the inequality $\log(1+x)\geq x/(1+x)$ and the fact that $\zeta\in(0,1)$, and without loss of generality, we set the initial discount factor $\gamma_0 = 1/(2\rho^2(A))$.
To elucidate the dependence on key system parameters $m, p, \rho(A)$, we suppress the dependence on user-specified parameters $\zeta$, $\delta_0$, and $\delta_1$. Then, the sample complexity can be expressed in the following compact form:
 \begin{equation}
     \mathcal{N}_{\text{total}} =\log(\rho(A))\cdot \mathcal{O}\left( \frac{m^2p^2}{\epsilon^4} \right)\cdot \text{poly}(\|A\|,\|B\|,\|Q\|,\|R\|, \overline{J}),
     \nonumber
 \end{equation}
where we use the fact that $\overline{G}$ and $\overline{M}$ are polynomials in the system parameters.

\section{Simulation Results}
\label{sec:simulation}
To validate our theoretical results regarding learning a stabilizing SOF policy for a discrete-time system, we perform two experiments.

\subsection{Numerical Example}
\label{sec:numerical_example}
Consider an unstable discrete-time system given by
\begin{equation}\label{eq-simulation}
    x_{t+1}
    =
    \begin{bmatrix}
        1.7 & 0.2 & 0\\
        0.2 & 0.7 & 0.1\\
        0 & 0.1 & 0.6
    \end{bmatrix}x_t
    +
    \begin{bmatrix}
        1\\
        0\\
        -0.1
    \end{bmatrix}u_t,
    \qquad
    y_t=
    \begin{bmatrix}
        1 & 0 & -0.1 \\
        -2 & 0 & 0.2
    \end{bmatrix}x_t.
\end{equation}
It can be verified that the above system matrices satisfy the conditions in Proposition~\ref{prop:symmetric-collocated-SIMO} and thus $\overline{J}$ in Theorem~\ref{theorem2} is finite as we argued in Section~\ref{Convergence Analysis}.

The open-loop spectral radius is $\rho(A)=1.739>1$, and hence the
system~\eqref{eq-simulation} is open-loop unstable. The cost matrices are selected as $Q=2\times I_3$ and $R=1$. The implementation of Algorithm~\ref{alg:learning_stabilize} uses the
following parameter settings:
$\gamma_0=0.2<1/\rho^2(A)=0.33$, $N=20$, $\zeta=0.8$, $\epsilon=1$, and $\tau_k=80$ for all $k\geq0$. The initial policy is chosen as $K_0=0$, and the constant stepsize for the PG update in~\eqref{eq-pg} is selected as $\eta_k=10^{-3}$ for all $k\geq0$. Furthermore, the gradient-estimation parameters in Algorithm~\ref{alg:two_point_grad} are chosen as
$\tau_k^e=60$, $r_k=10^{-2}$, and $N_k^e=40$ for all $k\geq0$.

\begin{figure}[htbp]
    \centering
    \subfloat[a][Closed-loop spectral radius of the system~\eqref{eq-simulation}.]{
    \includegraphics[width=0.4\textwidth]
    {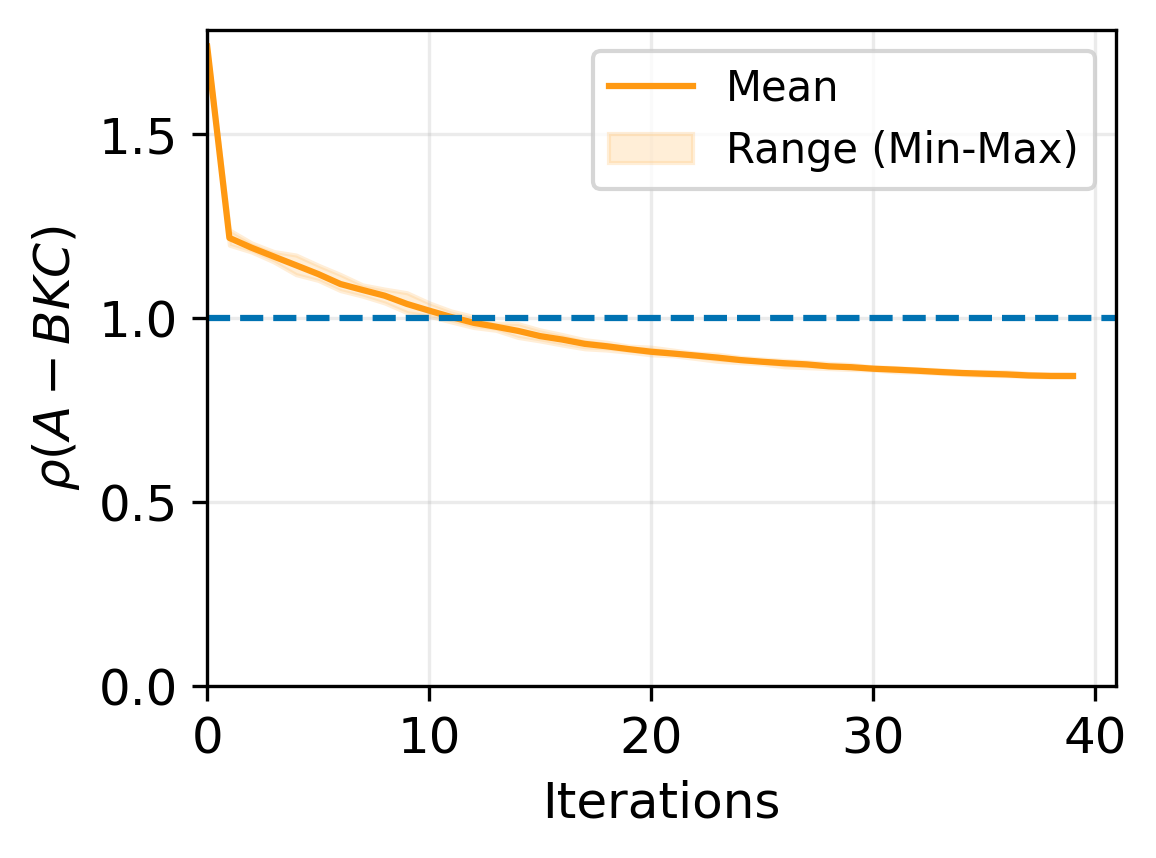}}
    \hfil
    \subfloat[b][Convergence of the discount factor $\gamma$.]{
    \includegraphics[width=0.4\textwidth]
    {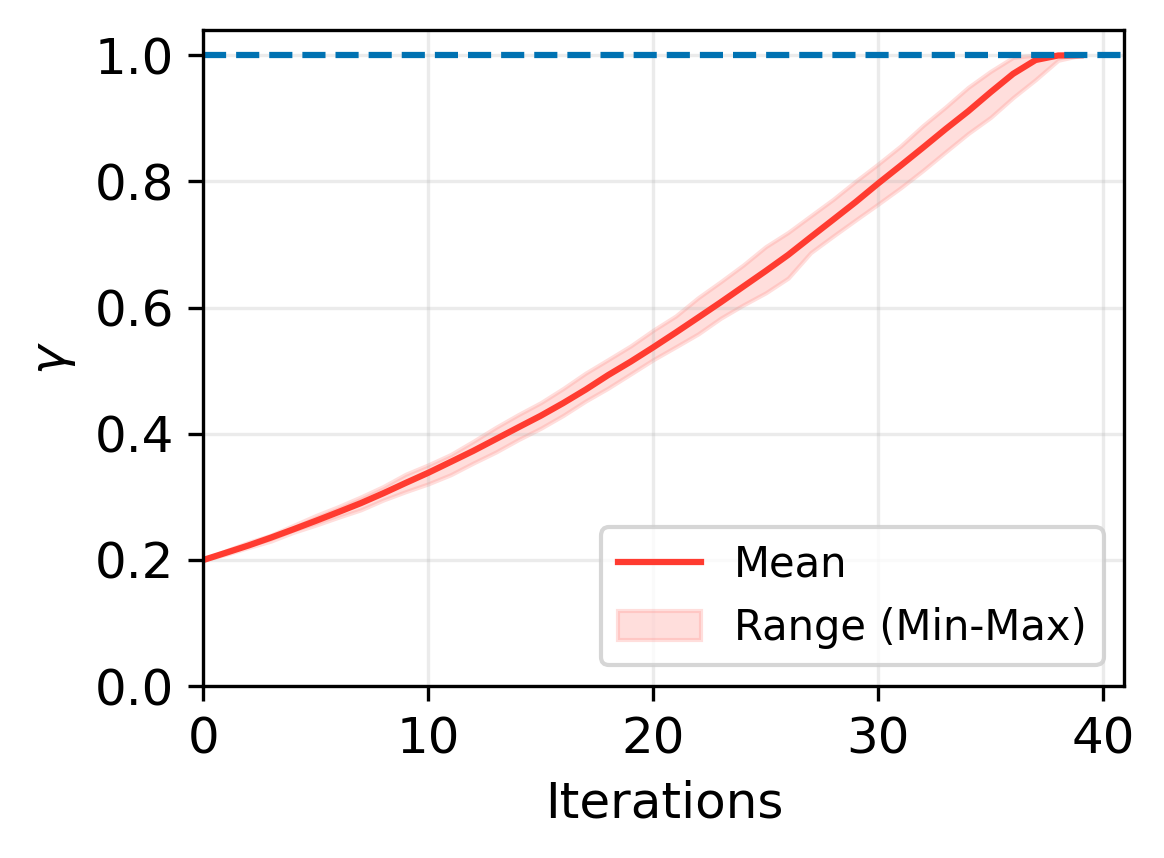}}
    \caption{Simulation results of Algorithm~\ref{alg:learning_stabilize} applied to the system~\eqref{eq-simulation} versus outer-loop iterations.}
    \label{fig-simulation}
\end{figure}

The solid lines in Figure~\ref{fig-simulation} represent the mean values over $20$ independent runs, while the shaded regions show the ranges between the minimum and maximum values. Figure~\ref{fig-simulation}(a) illustrates the evolution of the closed-loop spectral radius $\rho(A-BKC)$. Starting from the open-loop value $\rho(A)=1.739>1$, the spectral radius gradually decreases and eventually becomes smaller than one, indicating that the learned SOF controllers stabilize the original system~\eqref{eq-simulation}. Figure~\ref{fig-simulation}(b) shows that the discount factor gradually increases from $\gamma_0=0.2$ to $1$. These results demonstrate that Algorithm~\ref{alg:learning_stabilize} can learn a stabilizing SOF controller from the nonstabilizing initial policy $K_0=0$ using only system trajectories, without requiring knowledge of the system matrices $A$, $B$, and $C$.

\subsection{Application to Cart-Pole System}
\label{sec:application_cart_pole}
Let us consider the cart-pole dynamics~\cite{slotine1991applied} as a classic nonlinear and unstable  system, with its model given by:
\begin{equation}
\begin{aligned}
   &(m_c + m_p)\ddot{x}_c + m_p l_p \ddot{\theta} \cos\theta - m_p l_p \dot{\theta}^2 \sin\theta = u,\\
   & m_p l_p \ddot{x}_c \cos\theta + m_p l_p^2 \ddot{\theta} - m_p g l_p \sin\theta = 0,
\end{aligned}
    \nonumber
\end{equation}
where $x_c$ denotes the position of the cart, $\theta$ denotes the angle of the pendulum, and $u$ is the control input, representing the horizontal force applied to the cart.
The system parameters are set as follows: cart mass $m_c=0.5\,\text{kg}$, pendulum mass $m_p=2.0\,\text{kg}$, pole length $l_p=0.5\,\text{m}$, and gravitational acceleration $g=1.0\,\text{m}\cdot\text{s}^{-2}$.
By defining the state vector $x = [x_c, \theta, \dot{x}_c, \dot{\theta}]^\top$ and linearizing the system around the upright equilibrium, we obtain the continuous-time state-space model $\dot{x} = A_c x + B_c u$. Subsequently, applying Zero-Order Hold discretization with $T_s=0.1\,\text{s}$ yields the discrete-time model:
\begin{equation}\label{eq:discrete cartpole}
   x_{t+1}= 
    \begin{bmatrix}
       1 & 0.02 & 0.1 & 0\\
        0  & 1.05 & 0  &  0.1\\
        0  & 0.41 & 1 & 0.02\\
        0  & 1.02 & 0  &  1.05
    \end{bmatrix}x_t+ 
    \begin{bmatrix}
       0.01\\ 0.02\\ 0.2\\ 0.41
    \end{bmatrix}u_t.
\end{equation}
The system output is given by 
\begin{equation}
    y_t =  
    \begin{bmatrix}
       1 & 0 & 2 & 1 \\
       0 & 2 & 1 & 2 
    \end{bmatrix} x_t.
    \nonumber
\end{equation}
In this experiment, we set the cost matrices $Q=2\times I_4$ and $R=1$. The parameters in Algorithm~\ref{alg:two_point_grad} and Algorithm~\ref{alg:learning_stabilize} are chosen as follows: $\gamma_0=0.1<1/\rho(A)^2=0.533$, $N=20$, $\zeta=0.8$, $\epsilon=1$,  $\tau_k^e=\tau_k=100$, $r_k=10^{-2}$, $\eta_k=10^{-3}$, and $N^e_k=40$ for all $k\geq 0$. We perform $10$ independent runs and present the mean trajectory along with the minimum-maximum range in Figure~\ref{Fig2}.
Figure~\ref{Fig2}(a) illustrates the evolution of the closed-loop spectral radius $\rho(A-BKC)$ with respect to iterations. 
Figure~\ref{Fig2}(b) demonstrates the convergence of the discount factor $\gamma$. As can be observed, the discount factor increases from its initial value $\gamma_0=0.1$ to $1$ within 150 iterations.

\begin{figure}[htbp]
    \centering
    \subfloat[a][Closed-loop spectral radius of the system~\eqref{eq:discrete cartpole}.\label{Fig2(a)}]{    
    \includegraphics[width=0.4\textwidth]{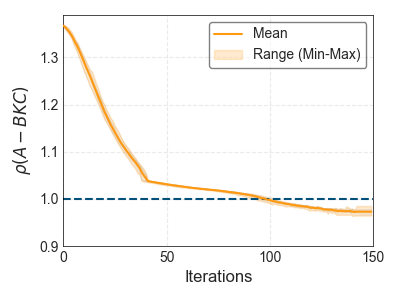}}
    \hfil
    \subfloat[b][Convergence of the discount factor $\gamma$.\label{Fig2(b)}]{    
    \includegraphics[width=0.4\textwidth]{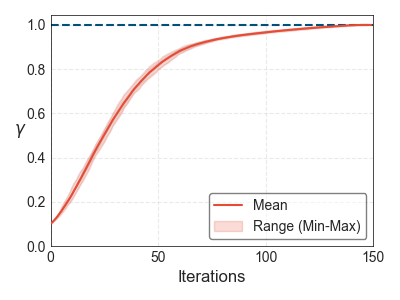}}
    \caption{Simulation results of Algorithm~\ref{alg:learning_stabilize} applied to the cart-pole system~\eqref{eq:discrete cartpole} versus iterations.}
    \label{Fig2}
\end{figure}

\section{Conclusion}\label{sec:conclusion}
In this paper, we have addressed the problem of learning stabilizing SOF controllers for partially observed discrete-time linear systems with unknown dynamics. While the optimization landscape of output feedback control is more complex than that of state feedback control, we propose a novel algorithmic framework based on the zeroth-order PG method, which only requires convergence to stationary points of the problem.
Theoretically, we established rigorous convergence guarantees of the proposed algorithm to an SOF controller for the unknown partially observed linear system, and provided a comprehensive analysis of the sample complexity. Numerical and application simulations further validated the efficacy and practicality of our framework. 
We list several future directions. First, extending to dynamic output-feedback controllers would require learning additional controller dynamics and analyzing a higher-dimensional and more complicated policy class, which is an important but nontrivial direction. 
Second, for stochastic linear systems, one may further consider noisy output measurements available to the controller. Finally, extensions to nonlinear control systems are more challenging, since our analysis relies on the linear system structure. We leave these extensions for future research.

\bibliography{references}
\bibliographystyle{unsrt}

\appendix
\counterwithin{lemma}{section}
\counterwithin{theorem}{section}
\counterwithin{proposition}{section}
\counterwithin{corollary}{section}
\counterwithin{definition}{section}
\counterwithin{equation}{section}
\onecolumn

\section{Proofs of Preliminary Results}
\label{app:preliminary proofs}
\subsection{Proof of Lemma~\ref{lem-pg-derivation}}
\label{Proof of lem-pg-derivation}
\begin{proof}
    We begin by considering the localized performance of policy $K$ via the quadratic value function $V_K(x_t) \coloneqq x_t^\top P_K x_t$, where $P_K$ is given by~\eqref{eq-Lya1}. Exploiting the recursive structure of the discounted Lyapunov equation \eqref{eq-Lya1}, the value function at the initial state can be expanded as:
    \begin{equation*}
        \begin{aligned}
            V_K(x_0) &= x_0^\top (Q + C^\top K^\top R K C) x_0 + \gamma V_K(x_1) \\
            &= x_0^\top (Q + C^\top K^\top R K C) x_0 + \gamma V_K((A - BKC)x_0).
        \end{aligned}
    \end{equation*}
    By taking the gradient of the above expression w.r.t.  $K$ (i.e., $\nabla_K V_K(x_0)$), we obtain the following recursive gradient relation:
    \begin{equation*}
        \begin{aligned}
            \nabla_K V_K(x_0) &= 2 R K C x_0 x_0^\top C^\top - 2 \gamma B^\top P_K (A - BKC) x_0 x_0^\top C^\top + \gamma \nabla_K V_K(x_1) \\
            &= 2 \left[ (R + \gamma B^\top P_K B)KC - \gamma B^\top P_K A \right] x_0 x_0^\top C^\top + \gamma \nabla_K V_K(x_1)\\
            &=2E_Kx_0 x_0^\top C^\top + \gamma \nabla_K V_K(x_1).
        \end{aligned}
    \end{equation*}
    Performing an infinite-horizon recursive expansion, the gradient at $x_0$ is expressed as follows:
    \begin{equation*}
        \nabla_K V_K(x_0) = 2 E_K \sum_{t=0}^{\infty} \gamma^t (x_t x_t^\top) C^\top.
    \end{equation*}
    Finally, by taking the expectation over the initial state distribution $\mathcal{D}$ and leveraging the fact that $J_{\gamma}(K) = \mathbb{E}_{x_0} V_K(x_0)$, we arrive at the closed-form expression $\nabla J_{\gamma}(K) = 2 E_K \Sigma_K C^\top$.
\end{proof}

\subsection{Proof of Lemma~\ref{lem1}}
\begin{proof}
The proof is divided into five parts corresponding to the statements in Lemma~\ref{lem1}.

\noindent\textbf{Part 1:}
    Since $K\in\mathcal{S}_{\gamma}(\nu)$, it follows from~\eqref{eq-Lya3} that $J_{\gamma}(K)=\Tr(P_K)\leq \nu$, and $J_{\gamma}(K)=\Tr((Q+C^\top K^\top RKC)\Sigma_K)\geq \underline{\sigma}(Q)\|\Sigma_{K}\|\geq \ell_0\|\Sigma_K\|$.

 \noindent\textbf{Part 2:}
    Let $F=KC$ and $F'=K'C$, and define two linear operators on a symmetric matrix $X$ as follows:
    \begin{equation*}
        \mathcal{T}_F(X):=\sum_{t=0}^{\infty}\gamma^t(A-BF)^tX[(A-BF)^\top ]^t,~ \mathcal{F}_F(X):=\gamma(A-BF)X(A-BF)^\top .
    \end{equation*}
 Meanwhile, define the induced norm of these operators as follows:
 \begin{equation*}
     \|\mathcal{T}_F\|=\sup_{X}\frac{\|\mathcal{T}_F(X)\|}{\|X\|},~\|\mathcal{F}_F\|=\sup_X\frac{\|\mathcal{F}_{F}(X)\|}{\|X\|}.
 \end{equation*}
Letting $X_0:=\mathbb{E}[x_0x_0^T]$, it follows from Assumption~\ref{ass1} that $X_0=I_n$. Notice that $\mathcal{T}_F(X_0)=\Sigma_{K}$, and so we get that
 $\|\mathcal{T}_F(X_0)\|=\|\Sigma_K\|\leq \nu/\ell_0=\kappa^2$.
From \cite[Lemma 17]{fazel2018global}, we obtain
$\|\mathcal{T}_F\|\leq J_{\gamma}(K)/(\underline{\sigma}(Q)\underline{\sigma}(X_0))\leq \nu/\ell_0=\kappa^2$,
and by \cite[Lemma 17]{fazel2018global}, we have that
\begin{equation} \label{eq-F-F'}
 \begin{aligned}
  \|\mathcal{F}_F-\mathcal{F}_{F'}\|\leq& 2\gamma\|B\|\|A-BF\|\|F-F'\|+\gamma\|B\|^2\|F-F'\|^2\\
   =&2\gamma\|B\|\|A-BKC\|\|KC-K'C\|+\gamma\|B\|^2\|KC-K'C\|^2\\
   \leq&2\gamma\|B\|\|C\|\|A-BKC\|\|K-K'\|+\gamma\|B\|^2\|C\|^2\|K-K'\|^2\\
   \leq&(2\gamma\psi\varphi\kappa+\frac{\gamma\psi^2\varphi^2}{8\kappa^3\psi\varphi})\|K-K'\|\\
   \leq&3\kappa\psi\varphi\|K-K'\|.
 \end{aligned}
\end{equation}
where the last inequality follows from $\gamma<1$ and $\kappa\geq1$.
By the proof of~\cite[Lemma 17]{fazel2018global}, if the following condition holds
\begin{equation}\label{eq-condition-1}
    \|F-F'\|\leq \frac{\underline{\sigma}(Q)\underline{\sigma}(X_0)}{4J_{\gamma}(K)\|B\|(\|A-BF\|+1)},
\end{equation}
then one has $ \|\Sigma_{K}-\Sigma_{K'}\|
    \leq2\|\mathcal{T}_F\|\|\mathcal{T}_F(X_0)\|\|\mathcal{F}_F-\mathcal{F}_{F'}\|.$
Since 
\begin{equation*}
    \frac{\underline{\sigma}(Q)\underline{\sigma}(X_0)}{4\nu\|B\|\|C\|(\|A-BKC\|+1)}
    \geq\frac{\ell_0}{8\nu\kappa\psi\varphi}
    =\frac{1}{8\kappa^3\psi\varphi}=D,
\end{equation*}
it follows from $\|K-K'\|\leq D$ and $F=KC$ that
\begin{equation}
 \begin{aligned}
    \|F-F'\|\leq\|K-K'\|\|C\|\leq \frac{\underline{\sigma}(Q)\underline{\sigma}(X_0)}{4\nu\|B\|(\|A-BF\|+1)},
 \end{aligned}
\end{equation}
which satisfies the condition~\eqref{eq-condition-1}. Thus, we get
\begin{equation}\label{eq-sigmaK-sigmaK'}
  \begin{aligned}
    \|\Sigma_K-\Sigma_{K'}\|
\leq2\|\mathcal{T}_F\|\|\mathcal{T}_F(X_0)\|\|\mathcal{F}_F-\mathcal{F}_{F'}\|
    \leq6\kappa^5 \psi\varphi\|K-K'\|,
  \end{aligned}
\end{equation}
where the last inequality follows from $\|\mathcal{T}_F\|\leq \kappa^2$, $\|\mathcal{T}_F(X_0)\|\leq \kappa^2$ and~\eqref{eq-F-F'}.

\noindent\textbf{Part 3:}
From~\eqref{eq-Lya1} we get that
\begin{equation}\label{eq-bound-P-P'}
    \begin{aligned}
        \|P_K-P_{K'}\|
        =&\|\mathcal{T}_F(Q+F^\top RF)-\mathcal{T}_{F'}(Q+F'^\top RF')\|\\
        =&\|\mathcal{T}_F(Q+F^\top RF)-\mathcal{T}_{F}(Q+F'^\top RF')-(\mathcal{T}_{F'}(Q+F'^\top RF')-\mathcal{T}_F(Q+F'^\top RF'))\|\\
        \leq&\|\mathcal{T}_F(F^\top RF-F'^\top RF')\|+\|\mathcal{T}_{F'}(Q+F'^\top RF')-\mathcal{T}_F(Q+F'^\top RF'))\|\\
        \overset{(a)}{\leq}&\|\mathcal{T}_F\|\|F^\top RF-F'^\top RF'\|+2\|\mathcal{T}_{F}\|^2\|\mathcal{F}_F-\mathcal{F}_{F'}\|\|Q+F'^\top RF'\|\\
        \leq&\|\mathcal{T}_F\|\|R\|\|F-F'\|^2+2\|\mathcal{T}_F\|\|R\|\|F\|\|F-F'\|
        +2\|\mathcal{T}_{F}\|^2\|\mathcal{F}_F-\mathcal{F}_{F'}\|(\|Q\|+\|R\|\|F'\|^2)\\
        \overset{(b)}{\leq}&\|\mathcal{T}_F\|\|R\|\|C\|^2\|K-K'\|^2+2\|\mathcal{T}_F\|\|R\|\|F\|\|C\|\|K-K'\|\\
        &+2\|\mathcal{T}_{F}\|^2\|\mathcal{F}_F-\mathcal{F}_{F'}\|
        \left[\|Q\|+\|R\|(D\|C\|+\|F\|)^2\right]\\
        \leq&\left[\kappa^2\ell_1\varphi^2D+2\kappa^2\ell_1\kappa\varphi
        +6\kappa^5\psi\varphi\ell_1(1+(D\varphi+\kappa)^2) \right]\|K-K'\|\\
        \leq&\kappa^5\ell_1\psi\varphi\left(8 D^2\varphi^2+16D\varphi\kappa+6+6D^2\varphi^2+12D\varphi+6\kappa^2\right)\|K-K'\|\\
        \overset{(c)}{\leq}&16\kappa^7\ell_1\psi\varphi\|K-K'\|,
    \end{aligned}
\end{equation}
where (a) follows from~\cite[Lemma 20]{fazel2018global}, (b) follows from the fact that $F=KC$, $F'=K'C$ and $\|F'\|\leq\|F-F'\|+\|F\|\leq D\|C\|+\|F\|$, and (c) is due to the fact that $\psi,\varphi,\kappa\geq1$, $\gamma<1$ and $D\varphi=1/8\kappa^3\psi\leq 1/8$. 

\noindent\textbf{Part 4:}
First, observing that $P_K \succeq Q \succeq \ell_0 I$, we deduce that $\Tr(X_0) \leq \ell_0^{-1}\Tr(P_K X_0) \leq \nu/\ell_0$. Combining this with the bound on $\|P_K-P_{K'}\|$ in~\eqref{eq-bound-P-P'}, we have
\begin{equation}
    \begin{aligned}
        |J_\gamma(K)-J_\gamma(K')|
        &= \left|\Tr(P_KX_0)-\Tr(P_{K'}X_0)\right|\\
        &\leq \|P_K-P_{K'}\| \Tr(X_0)\\
        &\leq \left(16\kappa^7\ell_1\psi\varphi\|K-K'\|\right) \frac{\nu}{\ell_0}\\
        &= 16\kappa^9\ell_1\psi\varphi\|K-K'\|.
        \nonumber
    \end{aligned}
\end{equation}

\noindent\textbf{Part 5:}
Recall that $\nabla J_{\gamma}(K)=2E_K\Sigma_KC^\top $, where $E_K=(R+\gamma B^\top P_KB)KC-\gamma B^\top P_KA$. To upper bound $\|\nabla J_{\gamma}(K)-\nabla J_{\gamma}(K')\|_F$, we first get 
\begin{equation}
    \nabla J_{\gamma}(K)-\nabla J_{\gamma}(K')=2E_K\Sigma_KC^\top -2E_{K'}\Sigma_{K'}C^\top =2\left[(E_K-E_{K'})\Sigma_{K'}+E_K(\Sigma_{K}-\Sigma_{K'})\right]C^\top.
    \nonumber
\end{equation}
Noticing that $\|P_{K'}\|\leq \|P_K\|+\|P_K-P_{K'}\|\leq 3\kappa^4\ell_1 $, we obtain
\begin{equation}\label{eq-E-E'}
    \begin{aligned}
        \|E_K-E_{K'}\|
        &\leq\|(R+\gamma B^\top P_KB)KC-\gamma B^\top P_KA-(R+\gamma B^\top P_{K'}B)K'C+\gamma B^\top P_{K'}A\|\\
        &=\|R(K-K')C-\gamma B^\top (P_K-P_{K'})(A-BKC)-\gamma B^\top P_{K'}B(K-K')C\|\\
        &\leq \|R\|\|C\|\|K-K'\|+\gamma\|B\|\|A-BKC\|\|P_K-P_{K'}\|+\gamma\|B\|^2\|C\|\|P_{K'}\|\|K-K'\|\\
        &\leq \left(\ell_1\varphi+16\kappa^8\ell_1\psi\varphi+3\kappa^4\ell_1\psi^2\varphi\right)\|K-K'\|\\
        &\leq20\kappa^8\ell_1\psi^2\varphi\|K-K'\|,
    \end{aligned}
\end{equation}
where the last inequality follows from $\kappa,\psi\geq1$.
Moreover, based on the first two statements in Lemma~\ref{lem1}, we have
\begin{equation} \label{eq-sigmak-E}
    \begin{aligned}
        &\|\Sigma_{K'}\|\leq\|\Sigma_{K}\|+\|\Sigma_K-\Sigma_{K'}\|\leq\nu/\ell_0+3\kappa^2/4\leq2\kappa^2,\\
        &\|E_{K}\|\leq\|R\|\|KC\|+\|B\|\|P_K\|\|A-BKC\|\leq\ell_1\kappa+\psi\nu\kappa\leq2\kappa^3\ell_1\psi.
    \end{aligned}
\end{equation}
Combining the above inequalities~\eqref{eq-sigmaK-sigmaK'},~\eqref{eq-E-E'} and~\eqref{eq-sigmak-E}, we obtain 
\begin{equation}
    \begin{aligned}
        \|\nabla J_{\gamma}(K)-\nabla J_{\gamma}(K')\|_F\leq&\sqrt{\min\{m,p\}} \|\nabla J_{\gamma}(K)-\nabla J_{\gamma}(K')\|\\
        \leq& 2\sqrt{\min\{m,p\}}(\|E_K-E_{K'}\|\|\Sigma_{K'}\|+\|E_K\|\|\Sigma_{K}-\Sigma_{K'}\|)\|C\|\\
        \leq&2\sqrt{\min\{m,p\}}(40\kappa^{10}\ell_1\psi^2\varphi+12\kappa^8\ell_1\psi^2\varphi)\|K-K'\| \\
        \leq&104\kappa^{10}\ell_1\psi^2\varphi^2\sqrt{\min\{m,p\}}\|K-K'\|.
    \end{aligned}
    \nonumber
\end{equation}
 This completes the proof.
\end{proof}

\subsection{Proof of Lemma~\ref{bound of gradient}}
\begin{proof}
     It immediately follows from $\nabla J_{\gamma}(K)=2E_K\Sigma_KC^\top $ that 
    \begin{equation}
    \begin{aligned}
         \|\nabla J_{\gamma}(K)\|_F
        &\leq\sqrt{\min\{m,p\}}\|\nabla J_{\gamma}(K)\|\\
        &\leq2\sqrt{\min\{m,p\}}\|E_K\|\|\Sigma_K\|\|C\|\\
        &\overset{(a)}{\leq}2\kappa^3\varphi(\ell_1+\psi\nu)\sqrt{\min\{m,p\}}=G^0
    \end{aligned}
    \nonumber
    \end{equation}
  where (a) is due to the bound on $\|E_K\|$ in~\eqref{eq-sigmak-E}, the bound on $\Sigma_K$ in Lemma~\ref{lem1} and the assumption that $\|C\|\leq\varphi$.
  This completes the proof.
\end{proof}

\section{Proofs omitted in Section~\ref{section 3}}
\subsection{Proof of Lemma~\ref{lem-gradient estimation error}}
\label{proof of lem-gradient estimation error}
In this section, we present the proof of Lemma~\ref{lem-gradient estimation error}. To this end, we first introduce an auxiliary function
    \begin{equation*}
        J^r_{\gamma}(K)=\mathbb{E}_{\widetilde{U}}[J_{\gamma}(K+r\widetilde{U})],~K\in\mathcal{S}_{\gamma}(\nu),
    \end{equation*}
     where the vectorized $\widetilde{U}$ is uniformly distributed from a unit ball. Accordingly, let the infinite-horizon version $\widetilde{\nabla}J_\gamma(K)$ of $\widehat{\nabla}J_\gamma(K)$ be given by
\begin{equation}
    \begin{aligned}
      \widetilde{\nabla}J_\gamma(K):=&\frac{mp}{2rN^e}\sum_{i=1}^{N^e}\left(J_{\gamma,x_0^i}(K+rU_i)-J_{\gamma,x_0^i}(K-rU_i)\right)U_i,
    \end{aligned}
\end{equation}
where $U_i\in\mathbb{R}^{m\times p}$ are random i.i.d. matrices and uniformly distributed from the set $\{U\mid\|U\|_F=1\}$. 
Then, we can decompose the gradient estimation error of $\widehat{\nabla}J_{\gamma}(K)$ as
    \begin{equation} \label{eq-error-all}
        \|\widehat{\nabla}J_{\gamma}(K)-\nabla J_{\gamma}(K)\|_F
        \leq \|\widehat{\nabla}J_{\gamma}(K)-\widetilde{\nabla}J_{\gamma}(K)\|_F
        +\|\widetilde{\nabla}J_{\gamma}(K)-\nabla J_{\gamma}^r(K)\|_F
        +\|\nabla J_{\gamma}^r(K)-\nabla J_{\gamma}(K)\|_F.
    \end{equation}
A crucial prerequisite for the two-point gradient estimation in Algorithm~\ref{alg:two_point_grad} is ensuring the stability of the perturbed closed-loop systems, i.e., $K\pm rU_i\in\mathcal{S}_{\gamma}$. Failure to satisfy this condition results in the exponential divergence of state and control signals. 
Moreover, $\nabla J_{\gamma}^r(K)$ and $\widetilde{\nabla}J_{\gamma}(K)$ are well-defined if and only if $K\pm rU_i\in\mathcal{S}_{\gamma}$. In Lemma~\ref{lem5}, we derive a condition on $r\leq r_0$, under which all perturbed policies $K\pm rU_i$ remain stabilizing.

\begin{lemma}\label{lem5}
    Let $K\in \mathcal{S}_{\gamma}(\nu)$. For any $K'\in\mathbb{R}^{m\times p}$ satisfying $\|K'-K\|\leq r_0$ with $r_0=\min\{D,\nu/G\}$, where $D$ and $G$ are defined in Lemma~\ref{lem1}, it holds that $K'\in\mathcal{S}_{\gamma}(2\nu)$.
\end{lemma}

\begin{proof}
    Since $r_0 \leq D$, invoking the local Lipschitz property of $J_{\gamma}(\cdot)$ from Lemma~\ref{lem1}, we have $|J_\gamma(K')-J_{\gamma}(K)|\leq G\|K'-K\|\leq Gr_0 \leq \nu$. This together with $J_{\gamma}(K)\leq\nu$ imply $J_{\gamma}(K')\leq 2\nu$, which completes the proof.
\end{proof}
    
Another crucial component of the proof is Lemma~\ref{norm(x_t)^2}, which establishes the exponential decay of $\|x_t\|^2$. This result serves as the SOF counterpart to~\cite[Lemma 9]{zhao2024convergence}.
\begin{lemma} \label{norm(x_t)^2}
    Consider the system described in~\eqref{eq-equal system}. Let $K\in \mathcal{S}_\gamma(\nu)$ and $P_K$ be a positive definite matrix that satisfies~\eqref{eq-Lya1}. For all $t\in\mathbb{N}$, it holds that $\|x_t\|^2\leq\left(1-\ell_0/\nu\right)^t\nu d^2/\ell_0$.
\end{lemma}

\begin{proof}
    Define the Lyapunov function $V(x_t)=x_t^\top P_Kx_t$. For the system $x_{t+1}=\sqrt{\gamma}(A-BKC)x_t$, it follows that
    \begin{equation}
    \begin{aligned}
         V(x_{t+1})-V(x_t)
         =& x_t^\top \gamma(A-BKC)^\top P_K(A-BKC)x_t-x_t^\top P_Kx_t \\
         \overset{(a)}{=}&-x_t^\top (Q+C^\top K^\top RKC)x_t\\
         \leq& -\frac{\underline{\sigma}(Q)}{\|P_K\|}V(x_t).
         \nonumber
    \end{aligned}
    \nonumber
    \end{equation}
    where (a) uses the Lyapunov equation~\eqref{eq-Lya1}. Since $V(x_t)\geq\underline{\sigma}(P_K)\|x_t\|^2$, we have 
    \begin{equation*}
        \|x_t\|^2\leq \frac{V(x_t)}{\underline{\sigma}(P_K)}\leq\left(1-\frac{\underline{\sigma}(Q)}{\|P_K\|}\right)^t\frac{V(x_0)}{\underline{\sigma}(P_K)}\leq\left(1-\frac{\underline{\sigma}(Q)}{\|P_K\|}\right)^t\frac{\text{Tr}(P_K)}{\underline{\sigma}(P_K)}\|x_0\|^2\leq \left(1-\frac{\ell_0}{\nu}\right)^t\frac{\nu d^2}{\ell_0},
    \end{equation*}
    where the last inequality is due to $\ell_0\leq \underline{\sigma}(Q)\leq\underline{\sigma}(P_K) $, $\|P_K\|\leq\Tr(P_K)\leq\nu$ and the assumption that $\|x_0\|\leq d$.
    This completes the proof.
\end{proof}

Based on this, the proof is divided into the following three parts.

\noindent\textbf{Bound on $\|\widehat{\nabla}J_{\gamma}(K)-\widetilde{\nabla}J_{\gamma}(K)\|_F$:}
    To upper bound $\|\widehat{\nabla}J_{\gamma}(K)-\widetilde{\nabla}J_{\gamma}(K)\|_F$ resulting from the finite simulation time, we first show that for $K\in\mathcal{S}_{\gamma}(\nu)$,
    \begin{equation}\label{eq-error of finite time}
       \begin{aligned}
           J_{\gamma,x_0}(K)-J_{\gamma,x_0}^{\tau^e}(K)
           =&\sum_{t=0}^{\infty}(x_t^\top Qx_t+u_t^\top Ru_t)-\sum_{t=0}^{\tau^e-1}(x_t^\top Qx_t+u_t^\top Ru_t)\\
           =&\sum_{t=\tau^e}^{\infty}x_t^\top (Q+C^\top K^\top RKC)x_t\\
           =&x_{\tau^e}^\top P_Kx_{\tau^e}\\
          \overset{(a)}{\leq}&\left(1-\frac{\ell_0}{\nu}\right)^{\tau^e}\frac{\nu d^2}{\ell_0}\Tr(P_K)\\
           \overset{(b)}{\leq}&\left(1-\frac{\ell_0}{\nu}\right)^{\tau^e} \frac{ \nu^2 d^2}{\ell_0},
       \end{aligned} 
    \end{equation}
    where (a) follows from Lemma~\ref{norm(x_t)^2}, and (b) uses the fact that $\Tr(P_K)=J_\gamma(K)\leq \nu$. 
    Given $\|U_i\|_F\leq1$ and the choice of $r$ satisfying $r\|U_i\|_F\leq r_0$, Lemma~\ref{lem5} implies that $K\pm rU_i\in\mathcal{S}_{\gamma}(2\nu)$.
    Therefore, we have the following upper bound on $\|\widehat{\nabla} J_\gamma(K)-\widetilde{\nabla} J_\gamma(K)\|_F$:
    \begin{equation}\label{eq-error-1}
        \begin{aligned}
            &\|\widehat{\nabla} J_\gamma(K)-\widetilde{\nabla} J_\gamma(K)\|_F\\
            \leq&\frac{mp}{2rN^e}\sum_{i=1}^{N^e}\left\|\left(J_{\gamma,x_0^i}(K+rU_i)-J_{\gamma,x_0^i}^{\tau^e}(K+rU_i)\right)U_i-\left(J_{\gamma,x_0^i}(K-rU_i)-J_{\gamma,x_0^i}^{\tau^e}(K-rU_i)\right)U_i\right\|_F\\
            \overset{(a)}{\leq}&\frac{mp}{2rN^e}\sum_{i=1}^{N^e}\left(\left|J_{\gamma,x_0^i}(K+rU_i)-J_{\gamma,x_0^i}^{\tau^e}(K+rU_i)\right|+\left|J_{\gamma,x_0^i}(K-rU_i)-J_{\gamma,x_0^i}^{\tau^e}(K-rU_i)\right|\right)\|U_i\|_F\\
            \overset{(b)}{\leq}&\frac{mp}{2rN^e}\cdot 2N^e\left(1-\frac{\ell_0}{2\nu}\right)^{\tau^e} \frac{ 4\nu^2 d^2}{\ell_0}\\
            =&\frac{4\nu^2 d^2mp}{r\ell_0}\left(1-\frac{\ell_0}{2\nu}\right)^{\tau^e} ,
        \end{aligned}
    \end{equation}
    where (a) follows from the triangle inequality, and (b) follows from  $\|U_i\|_F= 1$ and~\eqref{eq-error of finite time}, utilizing the fact that $K\pm rU_i\in\mathcal{S}_{\gamma}(2\nu)$.

    \noindent\textbf{Bound on $\|\widetilde{\nabla}J_{\gamma}(K)-\nabla J_{\gamma}^r(K)\|_F$:} To upper bound $\|\widetilde{\nabla}J_{\gamma}(K)-\nabla J_{\gamma}^r(K)\|_F$,
    we invoke \cite[Lemma 2.1]{flaxman2005online} to obtain
     \begin{equation*}
          \nabla J_{\gamma}^r(K)=\mathbb{E}_{U}[\frac{mp}{r}J_{\gamma}(K+rU)U],
     \end{equation*}
     where $U\in\mathbb{R}^{m\times p}$ are random matrices and uniformly distributed from the set $\{U\mid\|U\|_F=1\}$. Thus, we have 
     \begin{equation}
         \begin{aligned}
             \nabla J_{\gamma}^r(K)
              &=\mathbb{E}_{U}\left[\frac{mp}{r}J_{\gamma}(K+rU)U\right]\\
              &=\frac{mp}{2r}\left(\mathbb{E}_{U}[J_{\gamma}(K+rU)U]-\mathbb{E}_{U}[J_{\gamma}(K-rU)U]\right)\\
              &\overset{(a)}{=}\frac{mp}{2r}\mathbb{E}_{U}\left[J_{\gamma}(K+rU)U-J_{\gamma}(K-rU)U\right],
              \nonumber
         \end{aligned}
         \nonumber
     \end{equation}
     where (a) follows from the symmetry of the uniform distribution on the unit sphere (i.e., $U$ and $-U$ share the same distribution). Furthermore, noting that $\mathbb{E}_U[J_{\gamma}(K)U] = J_{\gamma}(K)\mathbb{E}_U[U] = 0$, we can upper bound $\|\nabla J^r(K)\|$ as follows:
     \begin{equation}\label{eq-bound of nabla J^r}
         \begin{aligned}
             \|\nabla J_{\gamma}^r(K)\|&\leq\frac{mp}{2r}\left\|\mathbb{E}_{U}\left[J_{\gamma}(K+rU)U-J_{\gamma}(K-rU)U\right]\right\|\\
             &=\frac{mp}{2r}\left\|\mathbb{E}_{U}\left[\left(J_{\gamma}(K+rU)-J_{\gamma}(K)\right)U+(J_{\gamma}(K)-J_{\gamma}(K-rU))U\right]\right\|\\
             &\overset{(a)}{\leq}\frac{mp}{2r}\mathbb{E}_{U}\left[\left\|\left(J_{\gamma}(K+rU)-J_{\gamma}(K)\right)U+(J_{\gamma}(K)-J_{\gamma}(K-rU))U\right\|\right]\\
             &\overset{(b)}{\leq}\frac{mp}{2r}\mathbb{E}_{U}\left[\big(\left|J_{\gamma}(K+rU)-J_{\gamma}(K)\right|+\left|J_{\gamma}(K)-J_{\gamma}(K-rU)\right|\big)\|U\|_F\right]\\
             &\overset{(c)}{\leq}\frac{mp}{2r}\mathbb{E}_{U}\left[(Gr\|U\|_F + Gr\|U\|_F)\|U\|_F\right]\\
             &= mpG,
         \end{aligned}
     \end{equation}
     where (a) follows from Jensen’s inequality and the convexity of $\|\cdot\|$; (b) applies the triangle inequality and the positive homogeneity of $\|\cdot\|$; and (c) follows from the local Lipschitz property of $J_{\gamma}(\cdot)$ as established in Lemma~\ref{lem1}, given that $\|U\|_F=1$ and $r\|U\|_F = r \leq D$.

 With the shorthand notation $K^+=K+rU_i$ and $K^-=K-rU_i$, we define
 \begin{equation*}
     Y_i=\frac{mp}{2r}\left(J_{\gamma,x_0^i}(K^+)-J_{\gamma,x_0^i}(K^-)\right)U_i-\nabla J^r_{\gamma}(K).
 \end{equation*}
 Then, we can upper bound $\|Y_i\|_F$ as
  \begin{equation}
    \begin{aligned}
        \|Y_i\|_F &= \left\|\frac{mp}{2r}\left(J_{\gamma,x_0^i}(K^+)-J_{\gamma,x_0^i}(K^-)\right)U_i-\nabla J^r_{\gamma}(K)\right\|_F\\
        &\leq \frac{mp}{2r}\left\|\left(J_{\gamma,x_0^i}(K^+)-J_{\gamma,x_0^i}(K)\right)U_i+\left(J_{\gamma,x_0^i}(K)-J_{\gamma,x_0^i}(K^-)\right)U_i\right\|_F+\|\nabla J^r_{\gamma}(K)\|_F\\
        &\leq \frac{mp}{2r}\left(\left|{x_0^i}^\top P_{K^+}x_0^i-{x_0^i}^\top P_{K}x_0^i\right|+\left|{x_0^i}^\top P_{K}x_0^i-{x_0^i}^\top P_{K^-}x_0^i\right|\right)\|U_i\|_F + \|\nabla J^r_{\gamma}(K)\|_F\\
        &\leq \frac{mp}{2r}\|x_0^i\|^2\left(\|P_{K^+}-P_{K}\|+\|P_{K}-P_{K^-}\|\right)\|U_i\|_F + \|\nabla J^r_{\gamma}(K)\|_F\\
        &\overset{(a)}{\leq} \frac{mpd^2}{2r}\cdot 16\kappa^7\ell_1\psi\varphi \left(\|K^+-K\|_F+\|K-K^-\|_F\right)\|U_i\|_F + \|\nabla J^r_{\gamma}(K)\|_F\\
        &\overset{(b)}{\leq} mpd^2\cdot16\kappa^7\ell_1\psi\varphi + mpG\\
        &\overset{(c)}{\leq}mpG(d^2+1),
    \end{aligned}
    \nonumber
  \end{equation}
  For (a), we utilize the assumption that $\|x_0\|\leq d$, and invoke the third part of Lemma~\ref{lem1} based on $K\in\mathcal{S}_{\gamma}(\nu)$ and $\|K^{+}-K\|_F=\|K-K^{-}\|_F=r\|U_i\|_F\leq D$. 
  Inequality (b) utilizes the upper bound $\|\nabla J^r_{\gamma}(K)\|_F\leq mpG$ from \eqref{eq-bound of nabla J^r}, while (c) follows from the fact that $16\kappa^7\ell_1\psi\varphi\leq 16\kappa^9\ell_1\psi\varphi=G$ given $\kappa\geq1$.
 Consequently, noting that $\mathbb{E}[Y_i|K]=0$, we apply Lemma~\ref{lem-Azuma} to deduce that
    \begin{equation}\label{eq-error-3}
        \|\widetilde{\nabla}J_{\gamma}(K)-\nabla J_{\gamma}^r(K)\|_F
        =\frac{1}{N^e}\left\|\sum_{i=1}^{N^e}Y_i\right\|_F
        \leq \frac{mpG(d^2+1)}{\sqrt{N^e}}\sqrt{2\log\frac{15}{\delta_0}}
    \end{equation}
    holds with probability at least $1-\delta_0$. 
    
     \noindent\textbf{Bound on $\|\nabla J_{\gamma}^r(K)-\nabla J_{\gamma}(K)\|_F$:}
    Next, we derive the upper bound on $\|\nabla J_{\gamma}^r(K)-\nabla J_{\gamma}(K)\|_F$. To this end, we show that
    \begin{equation}\label{eq-error-2}
      \begin{aligned}
        \|\nabla J_{\gamma}^r(K)-\nabla J_{\gamma}(K)\|_F
        &=\|\mathbb{E}_{\widetilde{U}}[\nabla J_{\gamma}(K+r\widetilde{U})]-\nabla J_{\gamma}(K)\|_F\\
        &=\|\mathbb{E}_{\widetilde{U}}[\nabla J_{\gamma}(K+r\widetilde{U})-\nabla J_{\gamma}(K)]\|_F\\
        &\overset{(a)}{\leq}\mathbb{E}_{\widetilde{U}}[\|\nabla J_{\gamma}(K+r\widetilde{U})-\nabla J_{\gamma}(K)\|_F]\\
        &\overset{(b)}{\leq}Lr,
      \end{aligned}
   \end{equation}
    where (a) follows from Jensen's inequality and the convexity of the norm function. For (b), due to $\|\widetilde{U}\|_F\leq 1$ and $r\|\widetilde{U}\|_F \leq r \leq D$, it follows from the local smoothness property of $J_{\gamma}(\cdot)$, which is established in Lemma~\ref{lem1}.
    
    Combining the results from the three parts of our analysis, we are now positioned to present the upper bound on the gradient estimation error $\|\widehat{\nabla}J_{\gamma}(K)-\nabla J_{\gamma}(K)\|_F$.
    Substituting \eqref{eq-error-1}, \eqref{eq-error-3}, and \eqref{eq-error-2} into \eqref{eq-error-all}, we get that
    \begin{equation}\label{eq-error-all2}
        \begin{aligned}
         \|\widehat{\nabla}J_{\gamma}(K)-\nabla J_{\gamma}(K)\|_F
        \leq \frac{4\nu^2 d^2mp}{r\ell_0}\left(1-\frac{\ell_0}{2\nu}\right)^{\tau^e} +\frac{mpG(d^2+1)}{\sqrt{N^e}}\sqrt{2\log\frac{15}{\delta_0}}+Lr
        \end{aligned}
        \nonumber
    \end{equation}
    holds with probability at least $1-\delta_0$. To ensure $\|\widehat{\nabla}J_{\gamma}(K)-\nabla J_{\gamma}(K)\|_F\leq\epsilon/3$, it suffices to choose $\tau^e$, $r$, and $N^e$ such that
    \begin{equation*}
        \tau^e\geq\frac{1}{\log(1-\ell_0/(2\nu))}\log\frac{r\epsilon\ell_0}{36d^2\nu^2mp},~N^e\geq\frac{162(mpG)^2(d^2+1)^2}{\epsilon^2}\log\frac{15}{\delta_0},~r\leq\frac{\epsilon}{9L}.
    \end{equation*}
    Using the numerical inequality $\log(1+x) \leq x$, a sufficient condition for the trajectory length $\tau^e$ is $\tau^e\geq\frac{2\nu}{\ell_0} \log\left(\frac{36d^2\nu^2mp}{r\epsilon\ell_0}\right)$. With these parameter choices, the desired error bound is guaranteed, which completes the proof of the lemma.$\hfill\blacksquare$

\subsection{Proof of Lemma~\ref{lem8}}
\label{proof of lem8}

\begin{lemma}\label{lem-bound of rho}
    For any $K\in\mathcal{S}_{\gamma}$, we have the spectral radius bound $\sqrt{\gamma}\rho(A-BKC)\leq\sqrt{1-\ell_0/J_{\gamma}(K)}$. Additionally, $\left\|( \sqrt{\gamma}(A-BKC))^t\right\|\leq\sqrt{J_{\gamma}(K)/\ell_0}\left(1-\ell_0/J_{\gamma}(K)\right)^{t/2}$ holds for all $t>0$.
\end{lemma}
\begin{proof}
    Since $K\in\mathcal{S}_{\gamma}$, it follows from~\eqref{eq-Lya2} that $\Sigma_K\succeq I_n$. On the other hand, from~\eqref{eq-Lya3}, we obtain
    \begin{equation}
        J_{\gamma}(K)=\Tr((Q+C^\top K^\top RKC)\Sigma_K)
                     \geq \Tr(\Sigma_K)\underline{\sigma}(Q+C^\top K^\top RKC)
                     \geq \Tr(\Sigma_K)\underline{\sigma}(Q)
                     \overset{(a)}{\geq} \ell_0\Tr(\Sigma_K),
                     \nonumber
    \end{equation}
    where (a) follows from the assumption that $\ell_0 I_n\preceq Q$.
    It implies that $\Sigma_K\preceq (J_{\gamma}(K)/\ell_0) I_n$.
    Next, pre- and post-multiplying~\eqref{eq-Lya2} by $\Sigma_K^{-1/2}$ yields
    \begin{equation}
    \begin{aligned}
         I_n&=\Sigma_K^{-1/2}\sqrt{\gamma}(A-BKC)\Sigma_K^{1/2}\left[\Sigma_K^{-1/2}\sqrt{\gamma}(A-BKC)\Sigma_K^{1/2}\right]^\top+\Sigma_K^{-1}\\
         &\succeq WW^\top+(\ell_0/J_{\gamma}(K))I_n,
    \end{aligned}
    \nonumber
    \end{equation}
    where we define $W:=\Sigma_K^{-1/2}\sqrt{\gamma}(A-BKC)\Sigma_K^{1/2}$. 
    Consequently, we obtain $WW^\top\preceq (1-\ell_0/J_{\gamma}(K))I_n$, which implies $\|W\|\leq\sqrt{1-\ell_0/J_{\gamma}(K)}$.
    Letting $H=\Sigma_K^{1/2}$, notice that $\sqrt{\gamma}(A-BKC)$ is similar to $W$, i.e., $\sqrt{\gamma}(A-BKC) = HWH^{-1}$. Thus,  we obtain
    $\sqrt{\gamma}\rho(A-BKC)=\rho(W)\leq\|W\|\leq\sqrt{1-\ell_0/J_{\gamma}(K)}$. Furthermore, it follows that
    \begin{equation}
    \begin{aligned}
        \left\|(\sqrt{\gamma}(A-BKC))^t\right\|
        &=\left\|HW^tH^{-1}\right\|\\
        &\leq \|H\|\|H^{-1}\|\|W\|^t\\
        &\leq \sqrt{\frac{J_{\gamma}(K)}{\ell_0}}\left(1-\frac{\ell_0}{J_{\gamma}(K)}\right)^{t/2}.
        \nonumber
    \end{aligned}
    \end{equation}
    The proof is completed.
\end{proof}

\begin{proof}[Proof of Lemma~\ref{lem8}]
Since $K\in\mathcal{S}_{\gamma}$, we obtain that
\begin{equation}
    \sqrt{\gamma'}\rho(A-BKC)\overset{(a)}{<}\sqrt{1+\alpha}\sqrt{\gamma}\rho(A-BKC)
    \overset{(b)}{\leq} \sqrt{1+\frac{\ell_0}{J_{\gamma}(K)-\ell_0}}\sqrt{1-\frac{\ell_0}{J_{\gamma}(K)}}=1,
    \nonumber
\end{equation}
where (a) follows from $\gamma'<(1+\alpha)\gamma$, (b) follows from~\eqref{eq-updating a} and Lemma~\ref{lem-bound of rho}.
It implies $K\in\mathcal{S}_{\gamma'}$.

\end{proof}

\section{Proofs Omitted in Theorem~\ref{theorem2}}
\label{Proofs Omitted in Theorem 2}
\subsection{Proof of Lemma~\ref{lem6}}
\begin{proof}
    Since $\tau_k$ and $N$ satisfy the conditions of Lemma~\ref{lem4}, the inequality $|J_{\gamma_k}(K_{k+1})-\widehat{J}_{\gamma_k}^{\tau_k,N}(K_{k+1})|\leq J_{\gamma_k}(K_{k+1})/2$ holds with probability at least $1-\delta_1$, which yields $J_{\gamma_k}(K_{k+1})\leq 2\widehat{J}_{\gamma_k}^{\tau_k,N}(K_{k+1})$.
    Due to $J_{\gamma_{k}}(K_{k+1})\leq \overline{J}$, we obtain
    $\sqrt{\gamma_k}\rho(A-BK_{k+1}C)<1$.
    Moreover, it follows from
    $\gamma_{k+1}=(1+\zeta \alpha_k)\gamma_k$ that
    \begin{equation}
        \gamma_{k+1}\overset{(a)}{<}(1+\alpha_k)\gamma_k
        = \left(1+\frac{\ell_0}{2\widehat{J}_{\gamma_k}^{\tau_k,N}(K_{k+1})-\ell_0}\right)\gamma_k
        \overset{(b)}{\leq} \left(1+\frac{\ell_0}{J_{\gamma_k}(K_{k+1})-\ell_0}\right)\gamma_k,
    \end{equation}
    where (a) uses the fact that $0<\zeta<1$, and (b) uses $J_{\gamma_k}(K_{k+1})\leq2\widehat{J}_{\gamma_k}^{\tau_k,N}(K_{k+1})$. Thus, we conclude that $\sqrt{\gamma_{k+1}}\rho(A-BK_{k+1}C)<1$ by Lemma~\ref{lem8}. 

\end{proof}

\subsection{Proof of Lemma~\ref{lem7}}
\begin{proof}
    We first prove the lower bound on $\alpha_k$. Since $|J_{\gamma_k}(K_{k+1})-\widehat{J}_{\gamma_k}^{\tau_k,N}(K_{k+1})|\leq J_{\gamma_k}(K_{k+1})/2$,
    we can derive an upper bound on $\widehat{J}_{\gamma_k}^{\tau_k,N}(K_{k+1})$ as follows:
    \begin{equation}
        \widehat{J}_{\gamma_k}^{\tau_k,N}(K_{k+1}) \leq |\widehat{J}_{\gamma_k}^{\tau_k,N}(K_{k+1})-J_{\gamma_k}(K_{k+1})|+ J_{\gamma_k}(K_{k+1}) \leq \frac{3}{2}J_{\gamma_k}(K_{k+1}).
        \nonumber
    \end{equation}
    Combined with $J_{\gamma_{k}}(K_{k+1})\leq \overline{J}$, we get $\widehat{J}_{\gamma_k}^{\tau_k,N}(K_{k+1}) \leq 3\overline{J}/2$, which yields
    \begin{equation}\label{eq-lem7-0}
        \alpha_{k}=\frac{\ell_0}{2\widehat{J}_{\gamma_k}^{\tau_k,N}(K_{k+1})-\ell_0}
        \geq \frac{\ell_0}{3\overline{J}-\ell_0}.
    \end{equation}
    
    To establish the upper bound on $J_{\gamma_{k+1}}(K_{k+1})$, we first recall its definition from~\eqref{eq-Lya3} and apply the properties of the trace
    \begin{equation}\label{eq-lem7-1}
       \begin{aligned}
        J_{\gamma_{k+1}}(K_{k+1})
        =&\Tr\left((Q+C^\top (K_{k+1})^\top RK_{k+1}C)\Sigma_{K_{k+1}}(\gamma_{k+1})\right)\\
        \leq &\Tr\left(Q+C^\top (K_{k+1})^\top RK_{k+1}C\right)\|\Sigma_{K_{k+1}}(\gamma_{k+1})\|,
       \end{aligned} 
    \end{equation}
    On the other hand, 
    we obtain
    \begin{equation}\label{eq-lem7-2}
        \begin{aligned}
            J_{\gamma_k}(K_{k+1})
            =&\Tr\left((Q+C^\top (K_{k+1})^\top RK_{k+1}C)\Sigma_{K_{k+1}}(\gamma_k)\right)\\
            \overset{(a)}{\geq}&\Tr(Q+C^\top (K_{k+1})^\top RK_{k+1}C),
        \end{aligned}
    \end{equation}
    where (a) follows from $\Sigma_{K_{k+1}}(\gamma_k) \succeq I_n$.
    Substituting \eqref{eq-lem7-2} into \eqref{eq-lem7-1}, one can show that
    \begin{equation}\label{eq-lem7-3}
         J_{\gamma_{k+1}}(K_{k+1}) \leq \Tr(Q+C^\top (K_{k+1})^\top RK_{k+1}C) \|\Sigma_{K_{k+1}}(\gamma_{k+1})\| 
        \leq J_{\gamma_{k}}(K_{k+1})  \|\Sigma_{K_{k+1}}(\gamma_{k+1})\|.
    \end{equation}
    Next, we establish the bound on $\left\|[\sqrt{\gamma_{k+1}}\left(A-BK_{k+1}C\right)]^t\right\|$ as follows:
    \begin{equation} \label{eq-lem7-4}
        \begin{aligned}
            &\left\|[\sqrt{\gamma_{k+1}}\left(A-BK_{k+1}C\right)]^t\right\|\\
            \overset{(a)}{=}&\left(\sqrt{1+\zeta\alpha_k}\right)^t
            \left\|[\sqrt{\gamma_k}(A-BK_{k+1}C)]^t\right\|\\
            \overset{(b)}{\leq}&\left(\sqrt{1+\zeta\alpha_k}\right)^t
            \sqrt{\frac{J_{\gamma_{k}}(K_{k+1})}{\ell_0}}\left(\sqrt{1-\frac{\ell_0}{J_{\gamma_k}(K_{k+1})}}\right)^t\\
            =& \sqrt{\frac{J_{\gamma_{k}}(K_{k+1})}{\ell_0}}
            \left(\sqrt{\frac{1+\zeta\alpha_k}{1+\alpha_k}}\right)^t\left[\sqrt{(1+\alpha_k)\left(1-\frac{\ell_0}{J_{\gamma_k}(K_{k+1})}\right)}\right]^t\\
            \overset{(c)}{\leq}&\sqrt{\frac{J_{\gamma_{k}}(K_{k+1})}{\ell_0}}
            \left(\sqrt{1-\frac{1-\zeta}{1+(\alpha_k)^{-1}}}\right)^t,
        \end{aligned}
    \end{equation}
    where (a) follows from the update rate for $\gamma_k$ designed in~\eqref{eq-updating a-new}, (b) uses Lemma~\ref{lem-bound of rho} based on $K_{k+1}\in\mathcal{S}_{\gamma_k}$, and (c) follows from 
    $(1+\alpha_k)(1-\ell_0/J_{\gamma_k}(K_{k+1}))\leq(1+\alpha_k)(1-\ell_0/2\widehat{J}_{\gamma_k}^{\tau_k,N}(K_{k+1}))\leq 1$.
    Recalling the definition of $\|\Sigma_{K_{k+1}}(\gamma_{k+1})\|$ from~\eqref{eq-Lya2}, we can obtain that
    \begin{equation} \label{eq-lem7-5}
        \begin{aligned}
            \|\Sigma_{K_{k+1}}(\gamma_{k+1})\|=&\left\|\sum_{t=0}^{\infty}(\gamma_{k+1})^t\left(A-BK_{k+1}C\right)^t\left[\left(A-BK_{k+1}C\right)^t\right]^\top \right\|\\
            \overset{(a)}{\leq}&\sum_{t=0}^{\infty}\left\|\left[\sqrt{\gamma_{k+1}}\left(A-BK_{k+1}C\right)\right]^t\right\|^2\\
            \overset{(b)}{\leq}&\frac{J_{\gamma_{k}}(K_{k+1})}{\ell_0}\sum_{t=0}^{\infty}\left(1-\frac{1-\zeta}{1+(\alpha_k)^{-1}}\right)^t\\
            =&\frac{J_{\gamma_{k}}(K_{k+1})}{\ell_0}\frac{1+(\alpha_k)^{-1}}{1-\zeta}\\
            \overset{(c)}{\leq}&\frac{4\left(\widehat{J}_{\gamma_k}^{\tau_k,N}(K_{k+1})\right)^2}{\ell_0^2(1-\zeta)},
        \end{aligned}
    \end{equation}
    where (a) uses the triangle inequality and the sub-multiplicativity of matrix norms, (b) follows from~\eqref{eq-lem7-4}, and (c) uses $J_{\gamma_k}(K_{k+1})\leq 2\widehat{J}_{\gamma_k}^{\tau_k, N}(K_{k+1})$.
    Inserting~\eqref{eq-lem7-5} into~\eqref{eq-lem7-3} and using $J_{\gamma_k}(K_{k+1})\leq 2\widehat{J}_{\gamma_k}^{\tau_k,N}(K_{k+1})$ completes the proof.
\end{proof}

\subsection{Proof of Lemma~\ref{lem12}}
Suppose that at iteration $k$, the conditions~\eqref{eq-lem12-1} and~\eqref{eq-lem12-2} hold. It follows from Lemma~\ref{lem7} that
\begin{equation}\label{eq-lem12-proof-1}
 J_{\gamma_{k+1}}(K_{k+1})\leq \frac{8\left(\widehat{J}_{\gamma_k}^{\tau_k,N}(K_{k+1})\right)^3}{(1-\zeta)\ell_0^2}\overset{(a)}{=}\nu_{k+1},
\end{equation}
where (a) follows from the choice of $\nu_{k+1}$ in Theorem~\ref{theorem2}.
The inequality~\eqref{eq-lem12-proof-1} confirms that $K_{k+1} \in \mathcal{S}_{\gamma_{k+1}}(\nu_{k+1})$. Consequently, based on the parameters $L_{k+1}, D_{k+1}, G_{k+1}^0, G_{k+1}$ in Lemma~\ref{lem1} and Lemma~\ref{bound of gradient}, which are polynomial functions of $\nu_{k+1}$, the selection of $\tau_{k+1}^e$, $N_{k+1}^e$, $r_{k+1}$, and $\eta_{k+1}$ in Theorem~\ref{theorem2} constitutes a sufficient condition for Theorem~\ref{theo1}.

  In addition, noting that 
  \begin{equation*}
      M_{k+1}=\frac{9\nu_{k+1}}{\eta_{k+1}\epsilon^2}
      \overset{(a)}{\geq}\frac{9J_{\gamma_{k+1}}(K_{k+1})}{\eta_{k+1}\epsilon^2},
  \end{equation*}
  where (a) is due to~\eqref{eq-lem12-proof-1},  we obtain that $M_{k+1}$ given in Theorem~\ref{theorem2} satisfies the conditions in Theorem~\ref{theo1}. 
  It follows from Theorem~\ref{theo1} that the zeroth-order PG method~\eqref{eq-pg} will return a stabilizing policy $K_{k+2}$ satisfying $J_{\gamma_{k+1}}(K_{k+2})\leq \overline{J}$ with probability at least $1-(M_{k+1}+1)\delta_0$. It further implies that $J_{\gamma_{k+1}}(K_{k+2})\leq J_{\gamma_{k+1}}(K_{k+1})\leq\nu_{k+1}$.

  Subsequently, we derive the success probability associated with $|J_{\gamma_{k+1}}(K_{k+2})-\widehat{J}_{\gamma_{k+1}}^{\tau_{k+1},N}(K_{k+2})|\leq J_{\gamma_{k+1}}(K_{k+2})/2$ under the assumption $J_{\gamma_{k+1}}(K_{k+2})\leq \nu_{k+1}$. 
  Indeed, applying $J_{\gamma_{k+1}}(K_{k+2})\leq \nu_{k+1}$ yields that
  \begin{equation*}
      \tau_{k+1}\geq\frac{2\nu_{k+1}}{\ell_0}\log\frac{\nu_{k+1}d^2}{\ell_0}
      \geq\frac{2J_{\gamma_{k+1}}(K_{k+2})}{\ell_0}\log\frac{J_{\gamma_{k+1}}(K_{k+2})d^2}{\ell_0}.
  \end{equation*}
  It is evident that $\tau_{k+1}$ and $N$ in~\eqref{eq-theo2-1} satisfy the conditions in~\eqref{eq-tau-N}. Invoking Lemma~\ref{lem4}, we obtain that $|J_{\gamma_{k+1}}(K_{k+2})-\widehat{J}_{\gamma_{k+1}}^{\tau,N}(K_{k+2})|\leq J_{\gamma_{k+1}}(K_{k+2})/2$ holds with probability at least $1-\delta_1$.
  Finally, by applying the union bound to the probabilities derived above, we conclude that at iteration $k+1$ both conditions~\eqref{eq-lem12-1} and~\eqref{eq-lem12-2} hold with probability at least $1 - \delta_1-(M_{k+1}+1)\delta_0$. This completes the proof.$\hfill\blacksquare$

\section{A Verifiable Class of  Single-Input Multi-Output Systems}
\label{sec:SIMO}
For any $\epsilon>0$, define the set of
$\epsilon$-stationary policies as 
\begin{equation*}
\mathcal{K}_{\epsilon}:=
\left\{(\gamma,K):\gamma\in[\gamma_0,1],\;
K\in\mathcal{S}_{\gamma},\;
\|\nabla_KJ_\gamma(K)\|_F\leq\epsilon\right\},
\end{equation*}
and define the corresponding family of effective gains (i.e., $KC$) as 
\begin{equation*}
\mathcal{E}_{\epsilon}
:=\left\{(\gamma,KC):
(\gamma, K)\in\mathcal{K}_{\epsilon}\right\}.
\end{equation*}
The discounted-stabilizing domain in the effective-gain space is defined as 
\begin{equation*}
\mathcal{D}
:=\left\{(\gamma,F): \gamma\in[\gamma_0,1], F=KC, K\in\mathcal{S}_{\gamma}\right\}.
\end{equation*}
Here, $\mathcal{D}$ is the set of all pairs $(\gamma,KC)$ with $\gamma\in[\gamma_0,1]$ and $K\in\mathcal{S}_{\gamma}$. In Theorem~\ref{theorem2}, we make the assumption that $\overline{\mathcal{E}_{\epsilon}}\subseteq\mathcal{D}$ and $\overline{\mathcal{E}_{\epsilon}}$ is compact. A sufficient condition for this assumption to hold is that there exist constants $h<\infty$ and $\delta>0$ such that every $(\gamma,K)\in\mathcal K_\epsilon$ satisfies
\begin{equation}\label{eq:equivalent_conditions}
\|KC\|_F\leq h,~
1-\gamma\rho(A-BKC)^2\geq\delta.
\end{equation}
To show the sufficiency, suppose the condition~\eqref{eq:equivalent_conditions} holds.  First, we show that $\overline{\mathcal E}_\epsilon$ is compact.
For any $(\gamma,F)\in\mathcal E_\epsilon$ with $F=KC$, it holds that
$\gamma\in[\gamma_0,1]$ and $\|F\|_F\le h$.
Hence, $\mathcal E_\epsilon$ is bounded, and so is its closure $\overline{\mathcal E}_\epsilon$. Since $\overline{\mathcal E}_\epsilon$ is closed by definition, it is compact
according to the Heine--Borel theorem~\cite{rudin2021principles}. We next show that $\overline{\mathcal E}_\epsilon\subseteq\mathcal D $.
Consider any $(\bar\gamma,\bar F)\in\overline{\mathcal E}_\epsilon$. We know that there
exists a sequence
$(\gamma_i,F_i)\in\mathcal E_\epsilon$ such that $(\gamma_i,F_i)\rightarrow(\bar\gamma,\bar F)$.
Since $1-\gamma_i\rho(A-BF_i)^2\ge\delta$
and the function $(\gamma,F)\mapsto 1-\gamma\rho(A-BF)^2$
is continuous, taking the limit gives $1-\bar\gamma\rho(A-B\bar F)^2\ge\delta>0$. Moreover, since $\gamma_i\in[\gamma_0,1]$,
we have $\bar\gamma\in[\gamma_0,1]$. The set $\{KC:K\in\mathbb R^{m\times p}\}$
is a finite-dimensional linear subspace and hence closed. Therefore, from $F_i=K_iC$ and $F_i\rightarrow\bar F$, we obtain $\bar F=\bar KC$ for some $\bar K$. Consequently, $(\bar\gamma,\bar F)\in\mathcal D$, which proves $\overline{\mathcal E}_\epsilon\subseteq\mathcal D$.

Proposition~\ref{prop:symmetric-collocated-SIMO} below identifies a verifiable class of single-input multi-output systems such that for every $\epsilon>0$, the condition~\eqref{eq:equivalent_conditions} holds with suitable constants $h<\infty$ and $\delta>0$, and hence the assumption above is satisfied. Consider a discrete-time single-input multi-output system
\begin{equation}
\label{eq:SIMO-system}
\begin{aligned}
x_{t+1}&=Ax_t+Bu_t,\\
y_t&=Cx_t,
\end{aligned}
\end{equation}
where $A\in\mathbb{R}^{n\times n}$, $B\in\mathbb{R}^{n\times 1}$,
and $C\in\mathbb{R}^{p\times n}$. The system is controlled by the
SOF control law $u_t=-Ky_t$, where $K\in\mathbb{R}^{1\times p}$.
The closed-loop matrix is denoted by $A_K:=A-BKC$.

For our analysis in the sequel, define the set 
\begin{equation}\nonumber
\Theta_{\mathcal B}
:=\left\{\theta\in\mathbb R:
1\leq\rho(A-\theta BB^\top)
\leq\gamma_0^{-1/2}\right\}.
\end{equation}
Throughout this section, assume that $n\geq2$ and
$\Theta_{\mathcal B}\neq\emptyset$. For a symmetric matrix $X\in\mathbb{R}^{n\times n}$, let $\lambda_1(X)\ge\lambda_2(X)\ge\cdots\ge\lambda_n(X)$ denote its eigenvalues, let 
$\sigma(X)$ denote its spectrum and define its peripheral spectrum as
\begin{equation}\nonumber
\mathcal{P}(X):=\left\{\mu\in\sigma(X):|\mu|=\rho(X)\right\}.
\end{equation}

\begin{proposition}
\label{prop:symmetric-collocated-SIMO}
Consider the single-input multi-output system \eqref{eq:SIMO-system}. Suppose that the following conditions hold.

\begin{enumerate}[leftmargin=*]
\item The pair $(A, B)$ is controllable, the matrix $A$ is symmetric, and there exists $a\in\mathbb{R}^{p}\setminus\{\mathbf{0}\}$  such that $C=a B^\top$.
\item The matrix $A$ has exactly one eigenvalue outside the unit disk, whereas all its remaining eigenvalues lie in $(-1,1)$.
\item The matrices $A$ and $B$ satisfy $B^\top(I-A^2)^{-1}B<0.$
\end{enumerate}
Then, for any $\epsilon>0$, there exist constants $h<\infty$ and $\delta>0$ such that, for every $(\gamma,K)\in\mathcal{K}_{\epsilon}$, we have $\|KC\|_F\leq h$ and $1-\gamma\rho(A-BKC)^2\geq\delta$.
\end{proposition}

\begin{proof}
\noindent\textbf{Step 1: Uniform boundedness of the effective gains.}
For every $\theta\in\Theta_{\mathcal B}$,
\begin{equation}\nonumber
\left|\Tr(A-\theta BB^\top)\right|
\leq
n\rho(A-\theta BB^\top)
\leq
n\gamma_0^{-1/2}.
\end{equation}
Since $\Tr(A-\theta BB^\top)=\Tr(A)-\theta B^\top B$,
we obtain
\begin{equation}\label{eq:theta-bound}
|\theta|
\leq
\frac{|\Tr(A)|+n\gamma_0^{-1/2}}
{\|B\|_2^2}
=:h_\theta.
\end{equation}
Therefore, $\Theta_{\mathcal B}$ is bounded.
Moreover, it is closed by the continuity of $\rho(A-\theta BB^\top)$ with respect to $\theta$, and hence $\Theta_{\mathcal B}$ is compact.
For every discounted-stabilizing pair $(\gamma,K)$ with
$\gamma\in[\gamma_0,1]$, let $\theta=Ka$. Since
\begin{equation}\nonumber
\rho(A-BKC)
=\rho(A-\theta BB^\top)
<\gamma^{-1/2}
\leq\gamma_0^{-1/2},
\end{equation}
the same trace argument as above gives $|\theta|\leq h_\theta.$
It follows from $C=aB^\top$ that
\begin{equation}\nonumber
\|KC\|_F=|Ka|\,\|B^\top\|_F
=|\theta|\,\|B\|_2\leq h_\theta\|B\|_2=:h.
\end{equation}

\noindent\textbf{Step 2: Spectral geometry and regularity of the stability boundary.}
It suffices to consider the case in which the unique unstable eigenvalue is greater than $1$. If $\lambda_n(A)<-1$, replace
$(A,B,C)$ with $(\bar A,\bar B,\bar C):=(-A,B,-C)$.
Since $\bar C=(-a)B^\top$, the collocated single-input multi-output structure is preserved. Then all the assumptions remain valid, 
\begin{equation}\nonumber 
\bar A-\bar B K\bar C=-(A-BKC), \quad B^\top(I-\bar A^2)^{-1}B = B^\top(I-A^2)^{-1}B, 
\end{equation}
and the discounted cost, its gradient, and the closed-loop spectral radius are unchanged. Moreover, $\bar A$ has exactly one eigenvalue greater than $1$, while all its remaining eigenvalues lie in $(-1,1)$. Therefore, without loss of generality, we assume that 
\begin{equation}\label{eq:wlog-positive-unstable-eigenvalue} \lambda_1(A)>1 \quad\text{and}\quad -1<\lambda_n(A)\leq\cdots\leq\lambda_2(A)<1. 
\end{equation}
Define $\theta:=Ka$ and $\widetilde A(\theta):=A-\theta BB^\top$. Since $C=aB^\top$, it follows that
\begin{equation}\nonumber
A_K=A-BKC=A-KaBB^\top=\widetilde A(\theta).
\end{equation}
Therefore, $A_K$ is symmetric for every
$K\in\mathbb{R}^{1\times p}$.

We first establish the spectral properties of the symmetric rank-one matrix family $\widetilde A(\theta)$. Suppose that $\widetilde A(\theta)$ has an eigenvalue $\lambda$ with algebraic multiplicity
greater than one. 
Since $\widetilde A(\theta)$ is symmetric, 
there exist two linearly independent eigenvectors $v_1$ and $v_2$ corresponding to $\lambda$. If $B^\top v_1=0$, we simply choose $v=v_1$. Otherwise, let
\begin{equation*}
v=v_2-\frac{B^\top v_2}{B^\top v_1}v_1.
\end{equation*}
Since $v_1$ and $v_2$ are linearly independent, $v\neq 0$. Moreover, $v$ is still an eigenvector corresponding to $\lambda$, and $B^\top v=0$ by construction.
It then follows that
$Av=\widetilde A(\theta)v+\theta BB^\top v
=\lambda v$.
Thus, $v$ is an eigenvector of $A$ satisfying $B^\top v=0$, which contradicts the PBH controllability condition for $(A, B)$. Therefore,
$\widetilde A(\theta)$ has simple eigenvalues for every $\theta\in\mathbb{R}$.
Let $\widetilde{\lambda}_i(\theta):=
\lambda_i\bigl(\widetilde A(\theta)\bigr),~ i=1,\ldots,n$, and let $v_i(\theta)$ be a corresponding unit eigenvector. The same
PBH argument gives $B^\top v_i(\theta)\neq0$
for every $i$ and every $\theta$. By the eigenvalue perturbation theorem~\cite{greenbaum2020first},
\begin{equation}
\label{eq:eigenvalue-monotonicity}
\frac{\mathrm{d}\widetilde{\lambda}_i(\theta)}{\mathrm{d}\theta}
=v_i(\theta)^\top
\frac{\mathrm{d}\widetilde{A}(\theta)}{\mathrm {d}\theta}v_i(\theta)
=-\big(B^\top v_i(\theta)\big)^2
<0.
\end{equation}
Hence, every eigenvalue of $\widetilde{A}(\theta)$ is strictly decreasing with respect to $\theta$.

Under~\eqref{eq:wlog-positive-unstable-eigenvalue}, we have $I+A\succ0$ and hence  $B^\top(-I-A)^{-1}B<0$. Moreover, 
\begin{equation}\nonumber B^\top(I-A)^{-1}B - B^\top(-I-A)^{-1}B = 2B^\top(I-A^2)^{-1}B<0. 
\end{equation} 
Define 
\begin{equation}\label{eq:theta-boundary-definitions} \theta_{+} := -\big[B^\top(I-A)^{-1}B\big]^{-1}, \quad \theta_{-} := -\big[B^\top(-I-A)^{-1}B\big]^{-1}. \end{equation} 
It then follows that $0<\theta_{+}<\theta_{-}$. 
For every $z\notin\sigma(A)$, the matrix determinant lemma~\cite{harville1998matrix} yields
\begin{equation}\nonumber
\det\bigl(zI-\widetilde A(\theta)\bigr)
=\det(zI-A)\big[1+\theta B^\top(zI-A)^{-1}B\big].
\end{equation}
Since $1\notin\sigma(A),-1\notin\sigma(A)$, the definitions of $\theta_+$ and $\theta_-$ imply that $1$
is an eigenvalue of $\widetilde A(\theta_+)$ and that $-1$ is an eigenvalue of $\widetilde A(\theta_-)$.
Since $\theta_+>0$, the eigenvalue interlacing inequalities for a
rank-one Hermitian perturbation~\cite{horn2012matrix} imply that $\widetilde{\lambda}_2(\theta_+)
\leq\lambda_2(A)<1$. It follows from $1\in\sigma\bigl(\widetilde A(\theta_+)\bigr)$ that $\widetilde{\lambda}_1(\theta_+)=1$.
Similarly, the eigenvalue interlacing inequalities give $\widetilde{\lambda}_{n-1}(\theta_-)
\geq\lambda_n(A)>-1$, which further yields $\widetilde{\lambda}_n(\theta_-)=-1$ since $-1\in\sigma\bigl(\widetilde A(\theta_-)\bigr)$.
By the strict monotonicity established in
\eqref{eq:eigenvalue-monotonicity}, for every
$\theta\in(\theta_+,\theta_-)$, we have $\widetilde{\lambda}_1(\theta)
<\widetilde{\lambda}_1(\theta_+)=1$ and $\widetilde{\lambda}_n(\theta)
>\widetilde{\lambda}_n(\theta_-)=-1$.
Therefore, $\rho\bigl(\widetilde A(\theta)\bigr)<1$ holds for $\theta\in(\theta_+,\theta_-)$.

We next characterize the peripheral spectrum outside this stabilizing
interval. If $\theta\leq\theta_+$, then
\begin{equation}\nonumber
\widetilde{\lambda}_1(\theta)\geq1,~
\widetilde{\lambda}_n(\theta)
\geq\widetilde{\lambda}_n(\theta_+)
>\widetilde{\lambda}_n(\theta_-)=-1.
\end{equation}
Since all eigenvalues of $\widetilde{A}(\theta)$ are simple, it follows
that $\mathcal{P}\bigl(\widetilde{A}(\theta)\bigr)
=\{\widetilde{\lambda}_1(\theta)\}$. 
Indeed, every negative eigenvalue has modulus strictly smaller than $1$, whereas every other nonnegative eigenvalue is strictly smaller
than $\widetilde{\lambda}_1(\theta)$.
Similarly, if $\theta\geq\theta_-$, then
\begin{equation}\nonumber
\widetilde{\lambda}_n(\theta)\leq-1,~
\widetilde{\lambda}_1(\theta)
\leq\widetilde{\lambda}_1(\theta_-)
<\widetilde{\lambda}_1(\theta_+)=1.
\end{equation}
Hence, $\mathcal{P}\bigl(\widetilde{A}(\theta)\bigr)
=\{\widetilde{\lambda}_n(\theta)\}$. 
For every $\theta\in\Theta_{\mathcal B}$, one has
$\rho(\widetilde A(\theta))\geq1$. Since $\widetilde A(\theta)$ is Schur stable for
$\theta\in(\theta_+,\theta_-)$,
we obtain $\theta\notin(\theta_+,\theta_-)$
for every $\theta\in\Theta_{\mathcal B}$. Consequently, $\widetilde A(\theta)$ has a unique, real, and simple peripheral eigenvalue for every
$\theta\in\Theta_{\mathcal B}$.
On the region $\theta\leq\theta_+$, the peripheral-spectrum characterization established above gives
$\rho\bigl(\widetilde A(\theta)\bigr)=
\widetilde{\lambda}_1(\theta)$.
By~\eqref{eq:eigenvalue-monotonicity},
$\widetilde{\lambda}_1(\theta)$ is strictly decreasing with respect to $\theta$. Similarly, on the region $\theta\geq\theta_-$, we have
$\rho\bigl(\widetilde A(\theta)\bigr)
=-\widetilde{\lambda}_n(\theta)$,
which is strictly increasing with respect to $\theta$, because $\widetilde{\lambda}_n(\theta)$ is strictly decreasing.
Therefore, on each of the two regions
$\theta\leq\theta_+$ and
$\theta\geq\theta_-$,
the set of scalar effective parameters satisfying
\begin{equation}\nonumber
1\leq\rho(\widetilde A(\theta))
\leq\gamma_0^{-1/2}
\end{equation}
is either empty or an interval.
Consequently, $\Theta_{\mathcal B}$ has at most two connected
components, corresponding to
$\theta\leq\theta_+$ and
$\theta\geq\theta_-$, respectively.
Each connected component is compact, since
$\Theta_{\mathcal B}$ is compact.

We now establish uniform transversality and spectral-gap estimates
on the relevant boundary parameters.
Denote the unique peripheral eigenvalue of
$\widetilde A(\theta)$ by $\lambda_{\mathrm p}(\theta)$, and let
$v_{\mathrm p}(\theta)$ be a corresponding unit eigenvector.
On each connected component of $\Theta_{\mathcal B}$,
$v_{\mathrm p}(\theta)$ can be chosen as a continuously
differentiable function of $\theta$.
By~\eqref{eq:eigenvalue-monotonicity},
\begin{equation}
\label{eq:boundary-transversality-symmetric}
\frac{\mathrm d}{\mathrm d\theta}
\lambda_{\mathrm p}(\theta)^2
=
-2\lambda_{\mathrm p}(\theta)
\left(B^\top v_{\mathrm p}(\theta)\right)^2
\neq0.
\end{equation}
We have shown that the peripheral eigenvalue of $\widetilde A(\theta)$ is unique and simple for
every $\theta\in\Theta_{\mathcal B}$. Moreover, the functions
\begin{equation}\nonumber
\theta\mapsto
\left|
\frac{\mathrm d}{\mathrm d\theta}
\lambda_{\mathrm p}(\theta)^2
\right|
\quad\text{and}\quad
\theta\mapsto
|\lambda_{\mathrm p}(\theta)|
-
\max_{\mu\in\sigma(\widetilde A(\theta))
\setminus\{\lambda_{\mathrm p}(\theta)\}}
|\mu|
\end{equation}
are continuous and strictly positive on each compact connected
component of $\Theta_{\mathcal B}$.
Since $\Theta_{\mathcal B}$ has at most two connected components,
there exist constants $c_0>0$ and $c_1>0$ such that
\begin{equation}
\label{eq:bound_on_lambda}
\left|
\frac{\mathrm d}{\mathrm d\theta}
\lambda_{\mathrm p}(\theta)^2
\right|
\geq c_0
\end{equation}
and
\begin{equation}\nonumber
|\lambda_{\mathrm p}(\theta)|
-
\max_{\mu\in\sigma(\widetilde A(\theta))
\setminus\{\lambda_{\mathrm p}(\theta)\}}
|\mu|
\geq c_1
\end{equation}
uniformly for $\theta\in\Theta_{\mathcal B}$.
 
\noindent\textbf{Step 3: Singular cost expansion and gradient growth near the boundary.}
Define
$\mathcal B
:=\big\{(\gamma,\theta):\gamma\in[\gamma_0,1],\;
\theta\in\Theta_{\mathcal B},\;\gamma\rho\bigl(A-\theta BB^\top\bigr)^2=1\big\}$. 
Since $\Theta_{\mathcal B}$ is compact and $\rho(A-\theta BB^\top)$ is continuous in $\theta$, $\mathcal B$ is a closed subset of the compact set
$[\gamma_0,1]\times\Theta_{\mathcal B}$.
Hence, $\mathcal B$ is compact.
Moreover, $\mathcal B$ is the graph of the continuous mapping $\theta\mapsto\rho(A-\theta BB^\top)^{-2}$
over $\Theta_{\mathcal B}$.
Therefore, $\mathcal B$ is homeomorphic to
$\Theta_{\mathcal B}$ and has at most two connected components.
Let $\mathcal B_1,\ldots,\mathcal B_k$, where $k\leq2$, denote the connected components of $\mathcal B$. On each component $\mathcal B_j$, the index of the unique peripheral eigenvalue is fixed; denote this index by $p_j\in\{1,n\}$.

Fix one component $\mathcal B_j$. For notational simplicity,
suppress the component index $j$ and write $p=p_j$.
For $\theta=Ka$, define $H(\theta):=Q+R\theta^2BB^\top$.
Since $C=aB^\top$ and $\theta=Ka$, we have
$H_K:=Q+C^\top K^\top RKC=H(\theta)$.
Let $\lambda_{\mathrm p}(\theta)
:=\widetilde\lambda_p(\theta)$, $\Pi_i(\theta):=v_i(\theta)v_i(\theta)^\top,~ i=1,\ldots,n$.
All the following definitions and estimates are componentwise. Since
$A_K=\widetilde A(\theta)$ is symmetric,
\begin{equation}\nonumber
A_K^t=\sum_{i=1}^{n}\widetilde\lambda_i(\theta)^t\Pi_i(\theta).
\end{equation}
Therefore, for every discounted-stabilizing pair $(\gamma,K)$
with $\theta=Ka$,
\begin{align}
J_\gamma(K)
=\sum_{t=0}^{\infty}\gamma^t\Tr\left(H(\theta)A_K^{2t}\right)
\nonumber
=\sum_{i=1}^{n}
\frac{\Tr\left(H(\theta)\Pi_i(\theta)\right)}{1-\gamma\widetilde\lambda_i(\theta)^2}.
\nonumber
\end{align}
Separating the term corresponding to $i=p$ gives
\begin{equation}
\label{eq:singular-cost-expansion-symmetric}
J_\gamma(K)
=\frac{g_1(\theta)}{1-\gamma\lambda_{\mathrm p}(\theta)^2}+g_2(\gamma,\theta),
\end{equation}
where $g_1(\theta)
:=\Tr\left(H(\theta)\Pi_p(\theta)\right)
=v_{\mathrm p}(\theta)^\top H(\theta)v_{\mathrm p}(\theta)$
and
\begin{equation}\nonumber
g_2(\gamma,\theta)
:=\sum_{\substack{1\leq i\leq n,~i\neq p}}
\frac{\Tr\left(H(\theta)\Pi_i(\theta)
\right)}{1-\gamma\widetilde\lambda_i(\theta)^2}.
\end{equation}
Since the preceding construction applies to every connected
component $\mathcal B_j$, we now restore the component index $j$.

We next justify the uniform boundedness of the terms
in~\eqref{eq:singular-cost-expansion-symmetric}. For each
$j=1,\ldots,k$, define the function $m_j$ on
$\mathcal B_j$ by
\begin{equation}\nonumber
m_j(\gamma,\theta)
:=\frac{\max_{i\neq p_j}|\lambda_i(A-\theta BB^\top)|
}{|\lambda_{p_j}(A-\theta BB^\top)|}.
\end{equation}
The uniqueness of the peripheral eigenvalue implies that
$m_j(\gamma,\theta)<1$ for every
$(\gamma,\theta)\in\mathcal B_j$.
Moreover, $m_j$ is continuous on $\mathcal B_j$. Since the sets $\mathcal B_j$ are compact and finite in number, there exists a constant $m_0<1$ such that
$m_j(\gamma,\theta)\leq m_0$
uniformly for all $j=1,\ldots,k$ and all
$(\gamma,\theta)\in\mathcal B_j$.
Choose a constant $\bar m$ satisfying $m_0<\bar m<1$,
and define $b:=1-\bar m^2>0$. 
Since the ordered eigenvalues $\widetilde\lambda_i(\theta)$ are continuous in $\theta$ and $\mathcal B_j$ is compact, the strict inequality $m_j(\gamma,\theta)<\bar m$ extends to a sufficiently small open neighborhood $\mathcal N_j$ of $\mathcal B_j$, with compact closure, such that
\begin{equation}\nonumber
\max_{i\neq p_j}\widetilde\lambda_i(\theta)^2
\leq (1-b)\widetilde\lambda_{p_j}(\theta)^2
\end{equation}
for every $(\gamma,\theta)\in\overline{\mathcal N_j}$.
The neighborhoods $\mathcal N_j$ may be chosen such that their closures are pairwise disjoint. Define
\begin{equation}\nonumber
\mathcal N
:=\bigcup_{j=1}^{k}\mathcal N_j.
\end{equation}
Whenever the neighborhoods $\mathcal N_j$ are subsequently shrunk, $\mathcal N$ is understood to be redefined as their union.
On each $\mathcal N_j$, we suppress the component index and write
$p=p_j$ and $\lambda_{\mathrm p}(\theta)=\widetilde\lambda_{p_j}(\theta)$. Since the closures $\overline{\mathcal N_j}$ are pairwise disjoint, the quantities $p$, $\lambda_{\mathrm p}$, $\Pi_p$, and $g_1$ are well defined componentwise on $\overline{\mathcal N}$, while $g_2$ is well defined for every discounted-stabilizing $(\gamma,\theta)\in\overline{\mathcal N}$.

Moreover, by the continuity of
\begin{equation}\nonumber
\theta\mapsto\frac{\mathrm d}{\mathrm d\theta}\lambda_{\mathrm p}(\theta)^2
\end{equation}
and the uniform lower bound in~\eqref{eq:bound_on_lambda}
on $\Theta_{\mathcal B}$, after shrinking each neighborhood
$\mathcal N_j$ if necessary, while preserving the preceding
spectral-gap estimate, we may define
$\bar c_0:=c_0/2>0$ and ensure that
\begin{equation}\nonumber
\left|\frac{\mathrm d}{\mathrm d\theta}\lambda_{\mathrm p}(\theta)^2\right|\geq\bar c_0
\end{equation}
for every
$(\gamma,\theta)\in\overline{\mathcal N}$.
Consequently, for every $i\neq p$ and every discounted-stabilizing
$(\gamma,\theta)\in\overline{\mathcal N}$, we have
$\gamma\widetilde\lambda_i(\theta)^2
\leq(1-b)\gamma\lambda_{\mathrm p}(\theta)^2<1-b$.
It follows that $1-\gamma\widetilde\lambda_i(\theta)^2
\geq b$ uniformly for every $i\neq p$.
In addition, after possibly shrinking each neighborhood
$\mathcal N_j$ further, while preserving the preceding estimates,
the eigenvalues $\widetilde\lambda_i(\theta)$ and the corresponding
spectral projectors $\Pi_i(\theta)$ are continuously differentiable
with respect to $\theta$ on each $\mathcal N_j$.
Since $\overline{\mathcal N}$ is compact, the functions
$\theta\mapsto H(\theta)$, $
\theta\mapsto\widetilde\lambda_i(\theta),
$ and $\theta\mapsto\Pi_i(\theta)$
together with their first derivatives with respect to $\theta$,
are uniformly bounded on $\overline{\mathcal N}$.

As a result, $g_1(\theta)$ and
$\mathrm d g_1(\theta)/\mathrm d\theta$ are uniformly bounded,
and both $g_2(\gamma,\theta)$ and
$\partial g_2(\gamma,\theta)/\partial\theta$ are uniformly bounded
over all discounted-stabilizing
$(\gamma,\theta)\in\overline{\mathcal N}$.
Finally, since
\begin{equation}\nonumber
H(\theta)
=Q+R\theta^2BB^\top
\succeq Q \succeq\ell_0I_n
\end{equation}
and $\|v_{\mathrm p}(\theta)\|_2=1$, we have
 $g_1(\theta)
=v_{\mathrm p}(\theta)^\top
H(\theta)v_{\mathrm p}(\theta)
\geq\ell_0.$

Define $d(\gamma,\theta)
:=1-\gamma\lambda_{\mathrm p}(\theta)^2$
componentwise on $\overline{\mathcal N}$.
Since $\theta=Ka$, differentiating
\eqref{eq:singular-cost-expansion-symmetric}
with respect to $K$ and applying the chain rule gives
\begin{equation}\nonumber
\nabla_KJ_\gamma(K)
=\left[\frac{\dfrac{\mathrm d}{\mathrm d\theta}g_1(\theta)}{
d(\gamma,\theta)}
+\frac{g_1(\theta)\gamma\dfrac{\mathrm d}{\mathrm d\theta}
\lambda_{\mathrm p}(\theta)^2}{d(\gamma,\theta)^2}
+\frac{\partial g_2(\gamma,\theta)}{\partial\theta}\right]a^\top.
\end{equation}
Since $\mathrm d g_1(\theta)/\mathrm d\theta$ and
$\partial g_2(\gamma,\theta)/\partial\theta$ are uniformly bounded over all discounted-stabilizing
$(\gamma,\theta)\in\overline{\mathcal N}$, there exist finite
constants $c_2,c_3>0$ such that
\begin{equation}\nonumber
\left|\frac{\mathrm d}{\mathrm d\theta}g_1(\theta)
\right|\leq c_2,
\qquad
\left|\frac{\partial}{\partial\theta}g_2(\gamma,\theta)
\right|\leq c_3
\end{equation}
uniformly over all discounted-stabilizing
$(\gamma,\theta)\in\overline{\mathcal N}$.
Therefore, for every discounted-stabilizing pair $(\gamma,K)$
satisfying $(\gamma,Ka)\in\overline{\mathcal N}$, with
$\theta=Ka$, the triangle inequality yields
\begin{equation}\nonumber
\|\nabla_KJ_\gamma(K)\|_F
\geq\|a\|_2\left(\frac{g_1(\theta)\gamma
\left|\dfrac{\mathrm d}{\mathrm d\theta}
\lambda_{\mathrm p}(\theta)^2
\right|}{d(\gamma,\theta)^2}
-\frac{c_2}{d(\gamma,\theta)}-c_3\right).
\end{equation}
Using
$g_1(\theta)\geq\ell_0$,
$\gamma\geq\gamma_0$, and
\begin{equation}\nonumber
\left|\frac{\mathrm d}{\mathrm d\theta}
\lambda_{\mathrm p}(\theta)^2\right|
\geq\bar c_0,
\end{equation}
we further obtain
\begin{equation}\label{eq:lower_bound_on_gradient}
\|\nabla_KJ_\gamma(K)\|_F
\geq\|a\|_2
\left(\frac{\ell_0\gamma_0\bar c_0}
{d(\gamma,\theta)^2}
-\frac{c_2}{d(\gamma,\theta)}-c_3\right).
\end{equation}

\noindent\textbf{Step 4: Uniform boundary separation and compact containment.}
First, we define
\begin{equation*}
    \mathcal C_\theta:=\left\{(\gamma,\theta):\gamma\in[\gamma_0,1],\;|\theta|\leq h_\theta\right\}.
\end{equation*}
By Step~1 above, every discounted-stabilizing pair $(\gamma,K)$ under
consideration satisfies $(\gamma,Ka)\in\mathcal C_\theta$.
Choose an open neighborhood $\mathcal U$ of $\mathcal B$ such that $\overline{\mathcal U}\subseteq\mathcal N.$
Define $q(\gamma,\theta):=1-\gamma\rho(A-\theta BB^\top)^2.$
If $q(\gamma,\theta)=0$, then $\rho(A-\theta BB^\top)
=\gamma^{-1/2}\in[1,\gamma_0^{-1/2}],$
and hence $\theta\in\Theta_{\mathcal B}$ and
$(\gamma,\theta)\in\mathcal B$.
Let $\mathcal C_0:=\left\{(\gamma,\theta)
\in\mathcal C_\theta\setminus\mathcal U:q(\gamma,\theta)\geq0
\right\}.$
Since $\mathcal C_\theta$ is compact, $\mathcal U$ is open, and
$q(\gamma,\theta)$ is continuous, $\mathcal C_0$ is compact.
Moreover, $\mathcal C_0$ contains no point satisfying
$q(\gamma,\theta)=0$, because every such point belongs to
$\mathcal B\subseteq\mathcal U$.
Consequently, $q(\gamma,\theta)>0$ on $\mathcal C_0$.
If $\mathcal C_0\neq\emptyset$, define $\tilde{\delta}
:=\frac{1}{2}\min_{(\gamma,\theta)\in\mathcal C_0}q(\gamma,\theta)>0.$
If $\mathcal C_0=\emptyset$, choose any
$\tilde{\delta}>0$.
It follows that every discounted-stabilizing pair $(\gamma,K)$
satisfying $0<1-\gamma\rho(A-BKC)^2 <\tilde{\delta} $
necessarily satisfies $(\gamma,Ka)\in\mathcal U.$

For every discounted-stabilizing pair $(\gamma,K)$ satisfying
$(\gamma,Ka)\in\mathcal U$, the unique
peripheral eigenvalue satisfies $\rho(A_K)
=\rho\bigl(\widetilde A(\theta)\bigr)
=|\lambda_{\mathrm p}(\theta)|,~ \theta=Ka,$
and hence $d(\gamma,\theta)=1-\gamma\lambda_{\mathrm p}(\theta)^2=1-\gamma\rho(A_K)^2.$
Define
\begin{equation}\nonumber
\phi(s)
:=\|a\|_2
\left(\frac{\ell_0\gamma_0\bar c_0}{s^2}-\frac{c_2}{s}
-c_3\right),\qquad s>0.
\end{equation}
Since $\phi(s)\to\infty$ as $s\to0^+$, for any
$\epsilon>0$, there exists
$\hat{\delta}>0$ such that $\phi(s)>\epsilon$ for every $0<s<\hat{\delta}$. 
Let $\delta=\min\{\tilde{\delta},\hat{\delta}\}$.
Then, for any discounted-stabilizing pair $(\gamma,K)$ satisfying
 $0<1-\gamma\rho(A_K)^2
<\delta$, we have
$(\gamma,Ka)\in\mathcal U\subseteq\mathcal N$,
and therefore $\|\nabla_KJ_\gamma(K)\|_F
\geq\phi\bigl(d(\gamma,Ka)\bigr)
>\epsilon$ by~\eqref{eq:lower_bound_on_gradient}.
Taking the contrapositive, we conclude that every
discounted-stabilizing pair satisfying
$\|\nabla_KJ_\gamma(K)\|_F\leq\epsilon$
also satisfies $1-\gamma\rho(A_K)^2\geq\delta$.

We have therefore shown that, for any $\epsilon>0$, there exist constants $h<\infty$ and $\delta>0$ such
that every discounted-stabilizing pair $(\gamma,K)$ satisfying $\|\nabla_KJ_\gamma(K)\|_F\leq\epsilon$
also satisfies $\|KC\|_F\leq h$ and $1-\gamma\rho(A-BKC)^2\geq\delta$.
Defining 
\begin{equation*}
\mathcal{C}_{h,\delta}
:=\left\{(\gamma,F):\gamma\in[\gamma_0,1],\;F=KC,\;K\in\mathbb{R}^{1\times p},\;\|F\|_F\leq h,\;1-\gamma\rho(A-BF)^2\geq\delta\right\},
\end{equation*}
the preceding implication yields
$\mathcal E_\epsilon\subseteq\mathcal C_{h,\delta}$.
Since $1-\gamma\rho(A-BF)^2$ is continuous in $(\gamma,F)$, $\mathcal C_{h,\delta}$ is closed. Moreover, since $\gamma\in[\gamma_0,1]$ and $\|F\|_F\leq h$ for every $(\gamma,F)\in\mathcal C_{h,\delta}$, $\mathcal C_{h,\delta}$ is bounded. By the Heine--Borel theorem~\cite{rudin2021principles}, $\mathcal C_{h,\delta}$ is compact, as it is a closed and bounded subset of the finite-dimensional space $\mathbb R\times\mathbb R^{1\times n}$. In addition, $\mathcal C_{h,\delta}\subseteq\mathcal D$, because every $(\gamma,F)\in\mathcal C_{h,\delta}$ satisfies $\gamma\rho(A-BF)^2\leq 1-\delta < 1$.
Since $\mathcal C_{h,\delta}$ is closed and contains $\mathcal E_\epsilon$, it also contains its closure. Consequently, $\overline{\mathcal E_\epsilon} \subseteq \mathcal C_{h,\delta} \subseteq\mathcal D$. Finally, $\overline{\mathcal E_\epsilon}$ is a closed subset of the compact set $\mathcal C_{h,\delta}$, and is therefore compact.
This completes the proof.
\end{proof}

A class of multi-input multi-output (MIMO) systems can also be considered by extending the setting in Proposition~\ref{prop:symmetric-collocated-SIMO}. 
Specifically, consider the system $x_{t+1}=Ax_t+Bu_t,~y_t=Cx_t$, which satisfies the following block-diagonal structure:
\begin{equation}\nonumber
A=\operatorname{blkdiag}(A_1,\ldots,A_m),~
B=\operatorname{blkdiag}(b_1,\ldots,b_m),~
C=\operatorname{blkdiag}
(a_1b_1^\top,\ldots,a_mb_m^\top),
\end{equation}
where each subsystem $i$ has $A_i\in\mathbb{R}^{n_i\times n_i}$, $b_i\in\mathbb{R}^{n_i}$ and $a_i\in\mathbb{R}\setminus\{0\}$. Further consider a block-diagonal feedback gain
$K=\operatorname{blkdiag}(k_1,\ldots,k_m)$. Then, the closed-loop matrix $A-BKC$ preserves the block-diagonal structure and can be decomposed into multiple independent SISO closed-loop systems, where each SISO block can be viewed as a special case of the single-input multi-output system considered in Proposition~\ref{prop:symmetric-collocated-SIMO}. In particular, each block has the same symmetric rank-one perturbation structure as that considered in Proposition~\ref{prop:symmetric-collocated-SIMO}. Thus, the spectral analysis and the stability-margin arguments in the proof of Proposition~\ref{prop:symmetric-collocated-SIMO} can be applied blockwise to this setting.

\section{Auxiliary Lemma}
\begin{lemma}\label{lem-Azuma}
    Let $\{Y_k\}_{k=1}^N$ be a martingale difference sequence with respect to a filtration $\{\mathcal{F}_k\}_{k=1}^N$, meaning that $\mathbb{E}[Y_k | \mathcal{F}_{k-1}] = 0$ for all $k$. Suppose that $\|Y_k\|_F\leq c$ for some constant $c$. Then, for any $N\geq1$ and any $\delta\in(0,e^{-2}/2)$, 
\begin{equation}
\Big\|\sum_{k=1}^NY_k\Big\|_F\leq c\sqrt{2N\log\frac{15}{\delta}}
\nonumber
\end{equation}
holds with probability at least $1-\delta$.
\end{lemma}
\begin{proof}
    The proof follows from~\cite[Theorem 1.8]{hayes2005large} directly.
\end{proof}
\end{document}